\documentclass[journal,twoside,web]{ieeecolor}

\def\logoname{logo}
\definecolor{subsectioncolor}{rgb}{0,0,0}
\usepackage{amsmath,amssymb,amsfonts, mathtools}
\usepackage{ntheorem}
\usepackage{siunitx}
\usepackage{bm}
\usepackage{graphicx}
\usepackage{url}
\usepackage{algorithm}
\usepackage{algcompatible}
\usepackage{mathrsfs}

\graphicspath{{pics/}}

\newcommand{\mat}[1]{\ensuremath{\begin{bmatrix} #1 \end{bmatrix}}}
\newcommand{\vc}[1]{\ensuremath{\begin{bmatrix} #1 \end{bmatrix}}}

\newcommand{\eps}{\varepsilon}
\renewcommand{\rho}{\varrho}
\renewcommand{\Pr}[1]{\mathrm{Pr}\left[#1\right]}
\newcommand{\E}[1]{\mathrm{E}\left[#1\right]}

\newcommand{\norm}[1]{\left\lVert #1 \right\rVert}
\newcommand{\norminf}[1]{\left\lVert #1 \right\rVert_{\infty}}

\renewcommand{\text}[1]{\textnormal{#1}}
\newcommand{\conv}[1]{\mathrm{conv}\left(#1\right)}
\newcommand{\col}[1]{\mathrm{col}\left(#1\right)}
\newcommand{\lag}[1]{\mathrm{lag}\left(#1\right)}
\newcommand{\evmax}[1]{\lambda_{\mathrm{max}}\left(#1\right)}
\newcommand{\evmin}[1]{\lambda_{\mathrm{min}}\left(#1\right)}
\newcommand{\svmax}[1]{\sigma_{\mathrm{max}}\left(#1\right)}
\newcommand{\trace}[1]{\mathrm{tr}\left(#1\right)}
\newcommand{\rank}[1]{\mathrm{rank}\left(#1\right)}

\newcommand{\conf}{\delta}
\newcommand{\hankel}[2]{\bm{H}_{#1}\hspace{-2pt}\left(#2\right)}

\newtheorem{theorem}{Theorem}
\newtheorem{lemma}{Lemma}
\newtheorem{proposition}{Proposition}
\newtheorem{remark}{Remark}

\newtheorem{definition}{Definition}
\newtheorem{assumption}{Assumption}
\newtheorem{problem}{Problem}

\usepackage{apptools}
\AtAppendix{
\renewtheorem{proposition}{Proposition}[subsection]
\renewtheorem{definition}{Definition}[subsection]
\renewtheorem{lemma}{Lemma}[subsection]
}

\usepackage[acronym, style=alttree, toc=true, shortcuts, xindy, nomain, nonumberlist]{glossaries}

\newglossary{symbols}{sym}{sbl}{List of Symbols}
\newglossary{notation}{not}{nt}{Notation}

\glsaddkey{dimension}{\glsentrytext{\glslabel}}{\glsentryd}{\GLsentryd}{\glsd}{\Glsd}{\GLSd}

\newacronym{MPC}{MPC}{model predictive control}
\newacronym{iss}{ISS}{input-to-state stability}

\newglossaryentry{x}{type=symbols,
	sort={x},
	dimension={\ensuremath{ n }},
	name={\ensuremath{\bm{x}}},
	description={State}
}

\newglossaryentry{xi}{type=symbols,
	sort={xi},
	dimension={\ensuremath{ n_{\xi} }},
	name={\ensuremath{\bm{\xi}}},
	description={Extended state}
}

\newglossaryentry{y}{type=symbols,
	sort={y},
	dimension={\ensuremath{ p }},
	name={\ensuremath{\bm{y}}},
	description={Output}
}

\newglossaryentry{u}{type=symbols,
	sort={u},
	dimension={\ensuremath{ m }},
	name={\ensuremath{\bm{u}}},
	description={Input}
}

\newglossaryentry{v}{type=symbols,
	sort={v},
	dimension={\ensuremath{ m }},
	name={\ensuremath{\bm{v}}},
	description={Input correction term}
}

\newglossaryentry{d}{type=symbols,
	sort={d},
	dimension={\glsd{y}},
	name={\ensuremath{\bm{d}}},
	description={Disturbance}
}

\newglossaryentry{var}{type=symbols,
	sort={var},
	dimension={\ensuremath{n_\zeta}},
	name={\ensuremath{\bm{\zeta}}},
	description={Deterministic decision variables}
}

\newglossaryentry{dset}{type=symbols,
	sort={dset},
	dimension={\glsd{y}},
	name={\ensuremath{\mathbb{D}}},
	description={Disturbance constraint set}
}

\newglossaryentry{ud}{type=symbols,
	sort={ud},
	dimension={\ensuremath{ \mathcal{U} }},
	name={\ensuremath{\bm{u}^{\mathrm{d}}}},
	description={Input data}
}

\newglossaryentry{vd}{type=symbols,
	sort={vd},
	dimension={\ensuremath{ \mathcal{V} }},
	name={\ensuremath{\bm{v}^{\mathrm{d}}}},
	description={Input correction term data}
}

\newglossaryentry{dd}{type=symbols,
	sort={dd},
	dimension={\ensuremath{ \mathcal{D} }},
	name={\ensuremath{\bm{d}^{\mathrm{d}}}},
	description={Disturbance data}
}

\newglossaryentry{xd}{type=symbols,
	sort={xd},
	dimension={\ensuremath{ \mathcal{X} }},
	name={\ensuremath{\bm{x}^{\mathrm{d}}}},
	description={State data}
}

\newglossaryentry{xid}{type=symbols,
	sort={xid},
	dimension={\ensuremath{ \mathcal{X} }},
	name={\ensuremath{\bm{\xi}^{\mathrm{d}}}},
	description={Extended state data}
}

\newglossaryentry{yd}{type=symbols,
	sort={yd},
	dimension={\ensuremath{ \mathcal{Y} }},
	name={\ensuremath{\bm{y}^{\mathrm{d}}}},
	description={Output data}
}

\newglossaryentry{Tini}{type=symbols,
	sort={Tini},
	dimension={\ensuremath{ 1 }},
	name={\ensuremath{T_{\mathrm{p}}}},
	description={Upper bound of lag}
}

\newglossaryentry{Tf}{type=symbols,
	sort={Tf},
	dimension={\ensuremath{ 1 }},
	name={\ensuremath{T_{\mathrm{f}}}},
	description={Prediction horizon}
}

\newglossaryentry{roa}{type=symbols,
	sort={\mathbb{C}_{\xi}^{\infty}},
	dimension={\ensuremath{\glsd{xi}}},
	name={\ensuremath{\mathbb{C}_{\xi}^{\infty}}},
	description={Region of Attraction}
}

\newglossaryentry{Nsaa}{type=symbols,
	sort={Nsaa},
	dimension={\ensuremath{ 1 }},
	name={\ensuremath{N_{\mathrm{s}}^{\mathrm{avg}}}},
	description={Number of samples for sample average approximation}
}

\usepackage{textcomp}
\def\BibTeX{{\rm B\kern-.05em{\sc i\kern-.025em b}\kern-.08em
    T\kern-.1667em\lower.7ex\hbox{E}\kern-.125emX}}
\begin{document}

\title{Sampling-based Stochastic Data-driven Predictive Control under Data Uncertainty -- Extended Version}

\author{
Johannes Teutsch, Sebastian Kerz, Dirk Wollherr, and Marion Leibold
\thanks{All authors are with the Chair of Automatic Control Engineering (LSR), Department of Computer Engineering, Technical University of Munich, Theresienstr. 90, 80333 Munich, Germany {\tt\small \{johannes.teutsch, s.kerz, dirk.wollherr , marion.leibold\}@tum.de}}
}

\maketitle

\begin{abstract}
    We present a stochastic constrained output-feedback data-driven predictive control scheme for linear time-invariant systems subject to bounded additive disturbances. The approach uses data-driven predictors based on an extension of Willems' fundamental lemma and requires only a single persistently exciting input-output data trajectory. Compared to current state-of-the-art approaches, we do not rely on availability of exact disturbance data. Instead, we leverage a novel parameterization of the unknown disturbance data considering consistency with the measured data and the system class. This allows for deterministic approximation of the chance constraints in a sampling-based fashion. A robust constraint on the first predicted step enables recursive feasibility, closed-loop constraint satisfaction, and robust asymptotic stability in expectation under standard assumptions. A numerical example demonstrates the efficiency of the proposed control scheme.

\end{abstract}


\section{Introduction}
Data-driven predictive control (DPC) promises safe and performant control of uncertain systems from measurement data \cite{coulson2019data,berberich2020data}.
As in model predictive control (MPC), DPC repeatedly solves a finite horizon optimal control problem (OCP), applying only the first input of the optimal input sequence at each time-step. 
The space of all finite length trajectories of a linear time-invariant (LTI) system is searched using a persistently exciting (PE) past input-output data trajectory based on Willems' fundamental lemma \cite{willems2005note}, and thus no explicit model is required.
In case additive disturbances affect the system, stochastic DPC leverages distributional information to guarantee the satisfaction of probabilistic chance constraints~\cite{pan2021stochastic, pan2023data, kerz2023datadriven}, similar to stochastic MPC \cite{mesbah2016}. This results in less conservative closed-loop behavior compared to robust data-driven control approaches, e.g., \cite{kloppelt2022novel}. 

Existing stochastic DPC schemes that come with closed-loop certificates for constraint satisfaction and stability are based on constraint tightening via chance-constrained optimization~\cite{kerz2023datadriven} or based on uncertainty propagation via polynomial chaos expansion~\cite{pan2021stochastic,pan2023data}.
However, persistently exciting measurement data of input-disturbance-state or input-disturbance-output trajectories are required, allowing for exact representation of the dynamics of the disturbed LTI system.
Other stochastic DPC schemes in literature are based on a similar assumption: In \cite{wang2022data}, a DPC scheme for stochastic systems in innovation form is presented, relying on available innovation data for predictions. Authors in \cite{yin2023stochastic} present a stochastic DPC scheme for unbounded noise, and \cite{li2023stochastic} show equivalence of stochastic DPC and MPC when data is exact.
While availability of input-state or input-output data is widely assumed in control theory, assuming that process noise can be measured or exactly estimated is restrictive, if not impractical.

In this work, we provide the first stochastic DPC scheme for LTI systems subject to process noise that only requires input-output data, while still guaranteeing recursive feasibility, satisfaction of chance constraints in closed-loop, and robust asymptotic stability in expectation (RASiE) \cite{mcallister2022nonlinear} with user-chosen confidence.
The key idea is to construct the set of disturbance trajectories consistent with the available input-output data in the sense that the combined input-disturbance-output data trajectory may have been produced by a disturbed LTI system. 
Given the disturbance distribution, this allows for sampling data-consistent disturbance trajectories with which data-driven predictors of the disturbed LTI system are constructed. These sample-based predictors are further used to deterministically approximate chance constraints via offline-sampling strategies \cite{mammarella2022chance,teutsch2023offline}, and an additional constraint on the first predicted step guarantees recursive feasibility and closed-loop constraint satisfaction \cite{lorenzen2017stochastic, mammarella2018offline}. 

The set of consistent disturbance trajectories corresponds to the set of consistent system parameters that arises from set membership identification \cite{milanese1991smi,teutsch2024adaptive} and that underpins the data informativity framework \cite{van2023informativity} and robust data-driven controller design \cite{berberich2020robustfb,alanwar2023data}.  
The switch of focus to consistent disturbances instead of model parameters lets us directly construct sample-based predictors without first mapping the disturbance distribution to a distribution over data-consistent system parameters for sample generation.
Moreover, compared to related sampling-based stochastic MPC schemes that come with stability certificates \cite{lorenzen2017stochastic,mammarella2018offline}, we do not assume that the expected finite-horizon cost is evaluated exactly. Instead, our stability analysis explicitly considers the approximation error resulting from sample average approximation (SAA) of the cost function via Hoeffding's inequality.
Our main contributions are as follows:

\begin{itemize}
    \item[C1:] We propose a novel parameterization of the unknown disturbance data considering consistency with the measured input-output data and the underlying system class (Proposition~\ref{prop:consistency_solutions}).
    \item[C2:] We present the first stochastic DPC scheme for disturbed LTI systems that guarantees chance constraint satisfaction and recursive feasibility without disturbance measurements (Algorithm~\ref{alg:controller}).
    \item[C3:] We provide a guarantee of RASiE with pre-specified confidence for the closed-loop system that explicitly considers the error resulting from SAA of the cost function (Theorem~\ref{th:rasie}).
\end{itemize}

This paper is an extended version of \cite{teutsch2025sampling} with additional material and discussions.

\subsubsection*{Organization}
We introduce the problem setup and discuss data-driven system representations in Section~\ref{sec:preliminaries}, and derive a parameterization of the unknown disturbance data in Section~\ref{sec:distdata}. The proposed controller is presented in Section~\ref{sec:method}, while its control-theoretic properties are discussed in Section~\ref{sec:properties}. Section~\ref{sec:eval} provides a numerical evaluation before we conclude the work in Section~\ref{sec:conclusion}.

\subsubsection*{Notation}
We write $\bm{0}$ for any zero matrix or vector and $\bm{I}_n$ for the identity matrix of order $n$. With $\bm{1}_n \in \mathbb{R}^n$, we denote a column-vector of all ones. We abbreviate the set of integers $\left\{a,\,\ldots,\, b\right\}$ by $\mathbb{N}_a^b$. The Moore-Penrose pseudo-rightinverse of a matrix $\bm{S}$ is defined as $\bm{S}^{\dagger} \coloneqq \bm{S}^{\top}\left(\bm{S}\bm{S}^{\top}\right)^{-1}$. The probability measure and the expectation operator are denoted as $\Pr{\cdot}$ and $\E{\cdot}$, respectively. The matrix $\left[\bm{S}\right]_{[a:b]}$ consists of all rows from the $a$-th row to the $b$-th row of the matrix $\bm{S}$, whereas $\left[\bm{S}\right]_{[a]}$ denotes the $a$-th row/element of the matrix/vector $\bm{S}$. The symbol ``$\otimes$" denotes the Kronecker product and ``$\le$" is applied element-wise. By $\col{\bm{s}_a,\,\ldots,\,\bm{s}_b} \coloneqq \mat{\bm{s}^{\top}_a\,\cdots\,\bm{s}^{\top}_b}^{\top}$, we denote the result from stacking the vectors/matrices $\bm{s}_a,\,\ldots,\,\bm{s}_b$. For any sequence of vectors $\mathcal{S}_T = \left\{\bm{s}_i\right\}_{i=1}^{T}$, $T\in\mathbb{N}$, the corresponding Hankel matrix $\hankel{L}{\mathcal{S}_T}$ of order $L \le T$ is
\begin{equation*}
\hankel{L}{\mathcal{S}_T} \coloneqq \mat{\bm{s}_1 & \bm{s}_2 & \cdots & \bm{s}_{T-L+1} \\ 
	\bm{s}_2 & \bm{s}_3 & \cdots & \bm{s}_{T-L+2} \\
	\vdots & \vdots & \ddots & \vdots \\
	\bm{s}_{L} & \bm{s}_{L+1} & \cdots & \bm{s}_{T}}.
\end{equation*}
For a positive definite matrix $\bm{S}$, we define the weighted 2-norm of the vector $\bm{s}$ as $\norm{\bm{s}}_{\bm{S}} \coloneqq \sqrt{\bm{s}^{\top} \bm{S} \bm{s}}$, and $\norm{\bm{s}} = \norm{\bm{s}}_{\bm{I}}$ for the Euclidean norm. We write $\bm{y}_{i|k}$ for the predicted output $i$ steps ahead of time-step $k$. For any sets $\mathbb{S}_1,\mathbb{S}_2$, we write the Minkowski set addition as $\mathbb{S}_1 \oplus \mathbb{S}_2=\{\bm{s}_1 + \bm{s}_2 \mid \bm{s}_1 \in \mathbb{S}_1,~ \bm{s}_2 \in \mathbb{S}_2\}$, the Pontryagin set difference as $\mathbb{S}_1 \ominus \mathbb{S}_2 = \{\bm{s}_1 \in \mathbb{S}_1 \mid \bm{s}_1 + \bm{s}_2 \in \mathbb{S}_1 ~\forall  \bm{s}_2 \in \mathbb{S}_2 \}$, and set multiplication as $\bm{K}\mathbb{S}_1 = \{\bm{K} \bm{s} \mid \bm{s} \in \mathbb{S}_1\}$. Positive definiteness of a matrix $\bm{S}$ is denoted by $\bm{S} \succ \bm{0}$, and $\conv{\cdot}$ denotes the convex hull over a set of vertices. We denote the maximum and minimum eigenvalue of a matrix $\bm{S}$ as $\evmax{\bm{S}}$ and $\evmin{\bm{S}}$, respectively. 
A function $\rho: \mathbb{R}_{\ge 0} \rightarrow \mathbb{R}_{\ge 0}$ is of class $\mathscr{K}$ if $\rho$ is continuous, strictly increasing, and $\rho(0) = 0$. If $\rho\in\mathscr{K}$ is unbounded, then $\rho$ is of class $\mathscr{K}_{\infty}$. A function $\beta: \mathbb{R}_{\ge 0} \times \mathbb{R}_{\ge 0} \rightarrow \mathbb{R}_{\ge 0}$ is of class $\mathscr{KL}$ if $\beta(\cdot,\,t) \in \mathscr{K}$ for fixed $t$ and $\beta(r,\,\cdot)$ is continuous, strictly decreasing, and ${\lim\limits_{t \rightarrow \infty}\beta(r,\,t)= 0}$ for fixed $r$.


\section{Problem Setup \& Preliminaries} \label{sec:preliminaries}
In this section, we first introduce the problem setup consisting of the considered system class and relevant assumptions. Then, we present preliminaries on data-driven system representations.

\subsection{Problem Setup}\label{sec:problemsetup}
We consider a discrete-time LTI system $\Sigma$ of order $n$ in AutoRegressive with eXtra input (ARX) form with additive disturbance, i.e.,
\begin{equation}
	\gls{y}_{k} = \bm{\Phi} \gls{xi}_{k} + \bm{\Psi} \gls{u}_{k} + \gls{d}_{k}, \label{eq:system}
\end{equation}
where the system matrices $\bm{\Phi}$, $\bm{\Psi}$ are unknown. System \eqref{eq:system} consists of the output $\gls{y}_{k}\in \mathbb{R}^{\glsd{y}}$, input $\gls{u}_{k} \in \mathbb{R}^{\glsd{u}}$, disturbance $\gls{d}_{k} \in \mathbb{R}^{\glsd{d}}$, and the \textit{extended state} vector of past $\gls{Tini} \in \mathbb{N}$ inputs and outputs
\begin{equation}
    \gls{xi}_{k} \coloneqq \vc{\col{\gls{u}_{k-\gls{Tini}},\,\dots,\,\gls{u}_{k-1}} \\ \col{\gls{y}_{k-\gls{Tini}},\,\dots,\,\gls{y}_{k-1}}} \in \mathbb{R}^{\glsd{xi}},\,\glsd{xi}\coloneqq(\glsd{u}+\glsd{y})\gls{Tini}, \label{eq:extstate}
\end{equation}
with given $\gls{xi}_{0}$. As in \cite{pan2023data}, we rely on the following assumption regarding an equivalent minimal state-space representation of \eqref{eq:system}.
\begin{assumption}[Minimal state-space representation]\label{assum:minss}
    There exists a minimal state-space representation 
    \begin{subequations} \label{eq:system_minimal}
    \begin{align}
    	\gls{x}_{k+1} &= \bm{A} \gls{x}_{k} + \bm{B} \gls{u}_{k} + \bm{E}\gls{d}_{k}, \\
    	\gls{y}_k &= \bm{C} \gls{x}_{k} + \bm{D} \gls{u}_{k} + \gls{d}_{k},
    \end{align}
    \end{subequations}
    with $\bm{x}\in\mathbb{R}^n$, controllable $\left(\bm{A},\,\mat{\bm{B} ~\bm{E}}\right)$, and observable $(\bm{A},\,\bm{C})$ such that for some initial condition $\gls{x}_{0}$, the input-output trajectories of \eqref{eq:system} and \eqref{eq:system_minimal} coincide for all disturbance sequences $\gls{d}_0, \gls{d}_1,\ldots$~.
\end{assumption}
Details on how to construct the system parameters in \eqref{eq:system_minimal} from (known) $\bm{\Phi}$ and $\bm{\Psi}$ in \eqref{eq:system} are given in \cite{pan2023data,sadamoto2022equivalence}. Assumption~\ref{assum:minss} allows us to construct a stabilizable and detectable \cite{bongard2022robust} (but not necessarily minimal) state-space representation of \eqref{eq:system}, i.e.,
\begin{subequations}\label{eq:system_arx_nonminimal}
\begin{align}
    \gls{xi}_{k+1} &= \tilde{\bm{A}} \gls{xi}_{k} + \tilde{\bm{B}} \gls{u}_{k} + \tilde{\bm{E}}\gls{d}_{k},\\
    \gls{y}_{k} &= \bm{\Phi} \gls{xi}_{k} + \bm{\Psi} \gls{u}_{k} + \gls{d}_{k}, 
\end{align}
\end{subequations}
with 
$\tilde{\bm{A}} \coloneqq \col{\bar{\bm{A}},\,\bm{\Phi}}$, $\tilde{\bm{B}} \coloneqq \col{\bar{\bm{B}},\,\bm{\Psi}}$, $\tilde{\bm{E}} \coloneqq \col{\bm{0},\, \bm{I}_{\glsd{y}}}$, where
\begin{equation}
    \bar{\bm{A}} \coloneqq \mat{\bm{0} & \bm{I}_{(\gls{Tini}-1)\glsd{u}} & \bm{0} & \bm{0} \\ \bm{0} & \bm{0} & \bm{0} & \bm{0} \\ \bm{0} & \bm{0} & \bm{0} & \bm{I}_{(\gls{Tini}-1)\glsd{y}} },\, \bar{\bm{B}} \coloneqq \mat{\bm{0}\\ \bm{I}_{\glsd{u}} \\ \bm{0}}. \label{eq:system_arx_nonminimal_helper}
\end{equation}
This equivalent state-space form of \eqref{eq:system} allows for simpler analysis of closed-loop properties, as in \cite{bongard2022robust}, and will be used in Section~\ref{sec:properties}.

Next to the assumptions on the system, we require that the disturbances $\gls{d}_k$ satisfy the following.
\begin{assumption}[Disturbance properties] \label{assum:distset}
    The disturbances $\gls{d}_k$ are the realizations of a zero-mean random variable that is independent and identically distributed (iid) according to a known probability density function $\bm{f}_{\gls{d}}(\cdot)$, supported by a known compact polytopic set
\begin{equation}\label{eq:distset}
    \gls{dset} = \left\{ \gls{d} \in \mathbb{R}^{\glsd{d}} ~\left|~ \bm{G}_d \bm{d} \le \bm{g}_d \right.\right\}.
\end{equation}
\end{assumption}

The problem considered in this work is to control system \eqref{eq:system} subject to probabilistic output constraints and hard input constraints 
\begin{subequations}\label{eq:constraints}
	\begin{align}
	&\Pr{\gls{y}_{k} \in \mathbb{Y}} \ge 1-\eps &\mathbb{Y} = \left\{\bm{y} \in \mathbb{R}^{\glsd{y}}  ~\left|~ \bm{G}_y \,\bm{y} \le \bm{g}_y \right.\right\}, \label{eq:outputcons} \\
	&\hspace{5mm}\gls{u}_{k} \in  \mathbb{U}, &\mathbb{U} = \left\{\bm{u} \in \mathbb{R}^{\glsd{u}} \left|~ \bm{G}_u \bm{u} \le \bm{g}_u \right.\right\}, \label{eq:inputcons}
	\end{align}
\end{subequations}
where $\mathbb{Y}$ and $\mathbb{U}$ are compact sets containing the origin. 
The objective of the predictive controller is to minimize in a receding horizon fashion the expected finite horizon cost
\begin{equation}\label{eq:cost_expected}
    J_{\gls{Tf}} \coloneqq \E{\sum\limits_{l=0}^{\gls{Tf}-1} \left( \norm{\gls{y}_{l|k}}^2_{\bm{Q}} + \norm{\gls{u}_{l|k}}^2_{\bm{R}} \right) + \norm{\gls{xi}_{\gls{Tf}|k}}^2_{\bm{P}}},
\end{equation}
with weights $\bm{Q}$, $\bm{R}$, $\bm{P} \succ \bm{0}$ (positive definite) and prediction horizon $\gls{Tf} \in \mathbb{N}$.
Since the system matrices $\bm{\Phi}$ and $\bm{\Psi}$ in \eqref{eq:system} are unknown in our problem setting, we cannot directly use \eqref{eq:system} for predictions. 
Instead, we base predictions on data, for which we assume access to a PE input-output data trajectory, collected before the control phase (\textit{offline}). Consider the following standard definition.
\begin{definition}[Persistency of excitation \cite{willems2005note}] \label{def:persistency}
	~~A trajectory $\mathcal{S}_T = \left\{\bm{s}_{i}\right\}_{i=1}^{T}$ of length $T\in\mathbb{N}$ with $\bm{s}_i \in \mathbb{R}^{n_s}$ is PE of order $L \le T$ if the Hankel matrix $\hankel{L}{\mathcal{S}_T}$ has full rank~$n_s L$.
\end{definition}
\begin{assumption}[Available data]\label{assum:trajData}
    An input-output data trajectory $\{\gls{ud}_i\}_{i=1-\gls{Tini}}^{T},\,\{\gls{yd}_i\}_{i=1-\gls{Tini}}^{T}$ of system \eqref{eq:system} is available, yielding the data
    $\glsd{ud}_T \coloneqq \{\gls{ud}_i\}_{i=1}^{T}$, $\glsd{yd}_T \coloneqq\{\gls{yd}_i\}_{i=1}^{T}$, and $\glsd{xid}_{T+1} \coloneqq \{\gls{xid}_i\}_{i=1}^{T+1}$ via \eqref{eq:extstate}. The corresponding disturbance data $\glsd{dd}_T \coloneqq \{\gls{dd}_{i}\}_{i=1}^{T}$ are unknown but satisfy Assumption~\ref{assum:distset} and they are such that the trajectory of generalized inputs $\{\col{\gls{ud}_i,\,\gls{dd}_i}\}_{i=1}^{T}$ is PE of order $\glsd{x}+\gls{Tf}+\gls{Tini}$, with system order~$\glsd{x}$ and horizon $\gls{Tf} \in \mathbb{N}$.
\end{assumption}
\begin{remark}
    Verifyability of Assumption~\ref{assum:trajData} is discussed in Appendix~\ref{app:verify}.
    Assumption~\ref{assum:trajData} is not restrictive in practice, since appropriate inputs $\gls{ud}_i$ can be chosen for the offline data collection, and $\gls{dd}_i$ is the realization of an iid random process (see Assumption~\ref{assum:distset}).
\end{remark}

Based on the available data, we aim to solve the following problem.
\begin{problem}\label{problem}
    Given a PE input-output data trajectory as in Assumption~\ref{assum:trajData}, design a computationally efficient output-feedback predictive control scheme for system \eqref{eq:system} subject to Assumption~\ref{assum:minss} that minimizes cost \eqref{eq:cost_expected} in receding horizon fashion while guaranteeing satisfaction of constraints \eqref{eq:constraints} during closed-loop operation. Constraint satisfaction entails the satisfaction of chance constraints~\eqref{eq:outputcons} based on probabilistic knowledge of the disturbance (Assumption~\ref{assum:distset}).
\end{problem}
We address Problem~\ref{problem} by developing a DPC scheme based on a data-driven system representation for model-free predictions. Although we assume that the system matrices in \eqref{eq:system} are fully unknown in general, we describe how prior (partial) model knowledge can be incorporated in the controller design in Appendix~\ref{app:example}.

\subsection{Data-driven System Representation} \label{sec:sysrep}
In their seminal work \cite{willems2005note}, Willems and co-authors have presented a non-parametric representation of LTI systems directly based on input-output data, known as the \textit{fundamental lemma}. 
The following result is an extension of said lemma to systems of the form~\eqref{eq:system}, cf. \cite{pan2021stochastic,kerz2023datadriven}.
\begin{lemma}[Extended fundamental lemma]
\label{lem:extfundamental}
	Consider a controllable LTI system $\Sigma$ of the form \eqref{eq:system_minimal} and measured data trajectories $\glsd{ud}_T$, $\glsd{dd}_T$, $\glsd{yd}_T$, and $\glsd{xid}_{T}$ where $T \ge \gls{Tini} + \gls{Tf}$, $\gls{Tf} \in \mathbb{N}$, and $\gls{Tini} \ge \lag{\Sigma}$.\footnote{$\lag{\Sigma}$ of an LTI system $\Sigma$ of order $n$ is defined as the smallest natural number $j \le n$ for which $\mathcal{O}_j \coloneqq \col{\bm{C}, \bm{C} \bm{A}, \dots, \bm{C} \bm{A}^{j-1}}$ has rank $n$.} 
    If the trajectory of generalized inputs $\{\col{\gls{ud}_i,\, \gls{dd}_i}\}_{i=1-\gls{Tini}}^{T}$ is PE of order $n+\gls{Tini} + \gls{Tf}$, then any length-$\left(\gls{Tini} + \gls{Tf}\right)$ input-disturbance-output trajectory $\left\{\gls{u}_{i}\right\}_{i=k-\gls{Tini}}^{k+\gls{Tf}-1}$, $\left\{\gls{d}_{i}\right\}_{i=k}^{k+\gls{Tf}-1}$, $\left\{\gls{y}_{i}\right\}_{i=k-\gls{Tini}}^{k+\gls{Tf}-1}$ is a valid trajectory of $\Sigma$ for $k \ge 0$ if and only if there exists $\bm{\alpha} \in \mathbb{R}^{T-\gls{Tf}+1}$ such that
    \begin{equation} \label{eq:extfundlemm_reform}
		\vc{\gls{xi}_k \\ \col{\gls{u}_{k},\,\dots,\,\gls{u}_{k+\gls{Tf}-1}}\\ \col{\gls{d}_{k},\,\dots,\,\gls{d}_{k+\gls{Tf}-1}}\\ \col{\gls{y}_{k},\,\dots,\,\gls{y}_{k+\gls{Tf}-1}}} = \mat{\hankel{1}{\glsd{xid}_{T-\gls{Tf}+1}} \\ \hankel{\gls{Tf}}{\glsd{ud}_{T}} \\ \hankel{\gls{Tf}}{\glsd{dd}_{T}} \\ \hankel{\gls{Tf}}{\glsd{yd}_{T}}} \bm{\alpha},
	\end{equation}
    with $\gls{xi}_k$ and $\glsd{xid}_{T-\gls{Tf}+1} = \{\gls{xid}_i\}_{i=1}^{T-\gls{Tf}+1}$ according to \eqref{eq:extstate}.
\end{lemma}
\begin{proof}
    The result follows directly from \cite[Lemma~1]{pan2021stochastic} by reordering the rows of \eqref{eq:extfundlemm_reform} and, due to the structure of \eqref{eq:system}, neglecting equations that involve the past disturbances $\gls{d}_{k-\gls{Tini}},\,\dots,\,\gls{d}_{k-1}$.
\end{proof}
Lemma~\ref{lem:extfundamental} lays the foundation for describing system behavior without model-knowledge in this work: equation \eqref{eq:extfundlemm_reform} functions as a non-parametric representation of system \eqref{eq:system_minimal}, and thus allows for the formulation of a data-driven OCP where \eqref{eq:extfundlemm_reform} replaces the prediction model and $\bm{\alpha}$ acts as a decision variable. Note that the extended state $\gls{xi}_k$ on the left-hand-side in \eqref{eq:extfundlemm_reform} implicitly fixes the initial state of the system for uniquely determined predictions since it entails the past $\gls{Tini}$ inputs and outputs $\left\{\gls{u}_{i}\right\}_{i=k-\gls{Tini}}^{k-1}$, $\left\{\gls{y}_{i}\right\}_{i=k-\gls{Tini}}^{k-1}$~\cite{markovsky2008}. 

Crucially, the system representation \eqref{eq:extfundlemm_reform} relies on the availability of disturbance data $\glsd{dd}_{T}$. Equivalent versions of Lemma~\ref{lem:extfundamental} are exploited in recent works for stochastic DPC, e.g., \cite{pan2023data,kerz2023datadriven}, where ${\glsd{dd}_{T}}$ is assumed to be known.
In this work, ${\glsd{dd}_{T}}$ is unknown as only input-output data is available (Assumption~\ref{assum:trajData}). In such a case, $\glsd{dd}_{T}$ may be estimated from inputs and outputs \cite{pan2021stochastic}. However, if the estimates are not exact, the guarantees of Lemma~\ref{lem:extfundamental} are lost: the right-hand side of \eqref{eq:extfundlemm_reform} might produce trajectories that are not realizable by the system.
In the next section, we address the issue of unknown disturbance data $\glsd{dd}_{T}$ by presenting an explicit parameterization considering consistency with the given input-output data and system class~\eqref{eq:system}.
\section{Consistency of Disturbance Data} \label{sec:distdata}

In this section, we will first derive an explicit parameterization of the unknown disturbance data $\glsd{dd}_T$ considering consistency with the given input-output data $\glsd{ud}_T$, $\glsd{yd}_T$. 
Then, we will derive a set of consistent disturbance data that can be further utilized for sampling.

\subsection{Consistent Disturbance Data}
With data from Assumption~\ref{assum:trajData}, let us consider the matrices
\begin{align*}
    \hankel{1}{\glsd{xid}_{T}} = \mat{\gls{xid}_1 & \cdots & \gls{xid}_T},\, \hankel{1}{\glsd{yd}_{T}} = \mat{\gls{yd}_1 & \cdots & \gls{yd}_T}, \\
    \hankel{1}{\glsd{ud}_{T}} = \mat{\gls{ud}_1 & \cdots & \gls{ud}_T},\, \hankel{1}{\glsd{dd}_{T}} = \mat{\gls{dd}_1 & \cdots & \gls{dd}_T}.
\end{align*}
As the given input-output data $\glsd{ud}_{T}$, $\glsd{yd}_{T}$ and unknown disturbance data $\glsd{dd}_{T}$ stem from system \eqref{eq:system}, the above data matrices must satisfy
\begin{equation}
	\hankel{1}{\glsd{yd}_{T}} = \bm{\Phi} \hankel{1}{\glsd{xid}_{T}} + \bm{\Psi} \hankel{1}{\glsd{ud}_{T}} + \hankel{1}{\glsd{dd}_{T}}. \label{eq:dynamics_data}
\end{equation}
Equation \eqref{eq:dynamics_data} allows for the definition of a constraint on the disturbance data $\glsd{dd}_{T}$ that guarantees consistency with the given input-output data $\glsd{ud}_{T}$, $\glsd{yd}_{T}$ and the underlying system class \eqref{eq:system}, see \cite{pan2021stochastic,berberich2020robustfb,alanwar2023data}. 
The following result summarizes \cite[Prop.~2~\&~Cor.~3]{pan2021stochastic}.
\begin{proposition}[Consistency constraint] \label{prop:consistency}
    Consider data $\glsd{ud}_{T}$, $\glsd{dd}_{T}$, $\glsd{yd}_{T}$ of system \eqref{eq:system} satisfying Assumption~\ref{assum:trajData} with $\glsd{dd}_{T}$ unknown. Then,
    \begin{equation}
        \left(\hankel{1}{\glsd{yd}_{T}} - \hankel{1}{\glsd{dd}_{T}}\right) \bm{\Pi}_S = \bm{0} \label{eq:consistency}
    \end{equation}
    holds, with $\bm{\Pi}_S \coloneqq \bm{I}_{T} - \bm{S}^{\dagger} \bm{S}$ and $\bm{S} \coloneqq \col{\hankel{1}{\glsd{xid}_{T}},\,\hankel{1}{\glsd{ud}_{T}}}$.

    Furthermore, any $\glsd{dd}_{T}$ satisfying \eqref{eq:consistency} implicitly determines parameters $\bm{\Phi}$, $\bm{\Psi}$ of an LTI system \eqref{eq:system} such that \eqref{eq:dynamics_data} is satisfied, i.e., 
\begin{equation}
    \mat{\bm{\Phi} & \bm{\Psi}} = \left(\hankel{1}{\glsd{yd}_{T}} - \hankel{1}{\glsd{dd}_{T}}\right) \bm{S}^{\dagger}. \label{eq:connection_sysparams_data}
\end{equation}
\end{proposition}
\begin{remark} \label{rem:rank}
Proposition~\ref{prop:consistency} requires full row-rank of the data matrix $\bm{S}$. This is a mild requirement since the input-output trajectory is randomly perturbed by disturbances at each time-step (see Appendix~\ref{app:verify}).
In a disturbance-free setting (or for small disturbance levels) and (nearly) rank-deficient $\bm{S}$, full row-rank of $\bm{S}$ can be recovered by an alternative definition of the extended state \cite{alsalti2023notes}.
\end{remark}
\begin{definition}[Consistent disturbance data trajectories]
Given input-output data as in Assumption~\ref{assum:trajData}, any disturbance trajectory $\glsd{dd}_{T}$ is called \emph{consistent} if and only if it satisfies~\eqref{eq:consistency}.
\end{definition}

Under the given assumptions, \eqref{eq:consistency} admits infinitely many solutions $\glsd{dd}_{T}$. The following proposition provides an explicit parameterization of these solutions in terms of $p(\glsd{xi}+\glsd{u})$ free parameters, namely $\glsd{xi}+\glsd{u}$ disturbances $\glsd{dd}_{\glsd{xi}+\glsd{u}}$ of the full disturbance data trajectory $\glsd{dd}_T$, specified via a column selection matrix $\bm{\Omega}\in \mathbb{R}^{T \times (\glsd{xi} + \glsd{u})}$.

\begin{proposition}[Consistency parameterization] \label{prop:consistency_solutions}
   Consider an input-output data trajectory $\glsd{ud}_{T}$, $\glsd{yd}_{T}$ of system \eqref{eq:system} satisfying Assumptions~\ref{assum:trajData} with $\glsd{dd}_{T}$ unknown. Let $\bm{\Omega}\in \mathbb{R}^{T \times (\glsd{xi} + \glsd{u})}$ be a column selection matrix that renders $\bm{S}\bm{\Omega}$ invertible and selects disturbances $\glsd{dd}_{\glsd{xi}+\glsd{u}}$ from $\glsd{dd}_T$ such that $\hankel{1}{\glsd{dd}_{\glsd{xi}+\glsd{u}}}=\hankel{1}{\glsd{dd}_{T}}\bm{\Omega}$. A candidate disturbance trajectory $\glsd{dd}_T$ is consistent if and only if it satisfies
    \begin{equation}
        \hankel{1}{\glsd{dd}_{T}} = \bm{\Gamma}_1 + \hankel{1}{\glsd{dd}_{\glsd{xi}+\glsd{u}}}\bm{\Gamma}_2 , \label{eq:consistency_solutions}
    \end{equation}
    where $\bm{\Gamma}_1$ and $\bm{\Gamma}_2$ are matrices computed from data as
    \begin{align} \label{eq:consistency_params}
        \bm{\Gamma}_1 &= \hankel{1}{\glsd{yd}_{T}} \bm{\Pi}_S\left(\bm{I}_T - \bm{\Omega}\bm{\Gamma}_2\right), & \bm{\Gamma}_2 &= \left(\bm{S} \bm{\Omega}\right)^{-1} \bm{S}.
    \end{align}
\end{proposition}
\begin{proof}
    Since $\bm{S} = \col{\hankel{1}{\glsd{xid}_{T}},\,\hankel{1}{\glsd{ud}_{T}}}$ has full row-rank $\glsd{xi}+\glsd{u}$ (see Remark~\ref{rem:rank}) and the rows of $\bm{S}$ span the null space of $\bm{\Pi}_S$, the solutions $\glsd{dd}_{T}$ of \eqref{eq:consistency} can be parameterized as
    \begin{equation}
        \hankel{1}{\glsd{dd}_{T}} = \hankel{1}{\glsd{yd}_{T}} \bm{\Pi}_S + \bm{\Delta} \bm{S}, 
    \label{eq:consistency_solutions_general}
    \end{equation}
    with the matrix $\bm{\Delta} \in \mathbb{R}^{\glsd{y} \times \left(\glsd{xi}+\glsd{u}\right)}$ containing the free parameters.
    We now want to express $\bm{\Delta}$ in terms of $\glsd{xi}+\glsd{u}$ disturbances from $\glsd{dd}_T$. Due to full row-rank of $\bm{S}$, there exists a column selection matrix $\bm{\Omega}$ such that $\bm{S} \bm{\Omega}$ is invertible. Thus, consider the $\glsd{xi}+\glsd{u}$ columns of \eqref{eq:consistency_solutions_general} according to $\bm{\Omega}$.
    With $\bm{S} \bm{\Omega}$ invertible, solving \eqref{eq:consistency_solutions_general} for $\bm{\Delta}$ yields
    \begin{equation}
        \bm{\Delta} = \left(\hankel{1}{\glsd{dd}_{\glsd{xi}+\glsd{u}}} -  \hankel{1}{\glsd{yd}_{T}} \bm{\Pi}_S \bm{\Omega}\right) \left(\bm{S} \bm{\Omega}\right)^{-1}. \label{eq:consistency_solution_delta}
    \end{equation}
    Finally, by substituting \eqref{eq:consistency_solution_delta} back into \eqref{eq:consistency_solutions_general}, we retrieve \eqref{eq:consistency_solutions} with the data-based parameters from \eqref{eq:consistency_params}.
\end{proof}

Under the consistency constraint \eqref{eq:consistency}, fixing $\glsd{xi}+\glsd{u}$ disturbances uniquely specifies the whole disturbance data trajectory of length~$T$, and thus implicitly specifying corresponding system parameters via \eqref{eq:connection_sysparams_data}. Equation \eqref{eq:consistency_solutions} thus offers a parameterization of all consistent disturbance data trajectories $\glsd{dd}_T$ based on $\glsd{xi}+\glsd{u}$ of the $T$ disturbances.
In the next section, we further restrict the choice of disturbances $\glsd{dd}_{\glsd{xi}+\glsd{u}}$ by including the disturbance bound \eqref{eq:distset} of Assumption \ref{assum:distset}.
\subsection{Set of Consistent Disturbance Data} \label{sec:distset_consistency}
Given the disturbance bound \eqref{eq:distset} in Assumption~\ref{assum:distset}, we are interested in all consistent disturbance trajectories $\glsd{dd}_{T}$ that are admissible, i.e., 
\begin{equation}
    \bm{G}_d \hankel{1}{\glsd{dd}_{T}} \le \bm{1}_{T}^{\top} \otimes \bm{g}_d\label{eq:distset_data_all}.
\end{equation}
By exploiting \eqref{eq:consistency_solutions}, we can express \eqref{eq:distset_data_all} in terms of $\glsd{dd}_{\glsd{xi}+\glsd{u}}$ as
\begin{align}\label{eq:distset_consistency}
    \gls{dset}^{\mathrm{c}} =
    \left\{ \glsd{dd}_{\glsd{xi}+\glsd{u}} \left| 
      \bm{G}_d \hankel{1}{\glsd{dd}_{\glsd{xi}+\glsd{u}}}\bm{\Gamma}_2 \le \bm{G}_d^{\mathrm{c}}  \right.\right\},
\end{align}
with $\bm{G}_d^{\mathrm{c}} \coloneqq \bm{1}_{T}^{\top} \otimes \bm{g}_d - \bm{G}_d\bm{\Gamma}_1$.
If $\glsd{dd}_{\glsd{xi}+\glsd{u}}\in \gls{dset}^{\mathrm{c}}$, then the associated disturbance trajectory $\glsd{dd}_{T}$ \eqref{eq:consistency_solutions} is consistent and satisfies the bounds \eqref{eq:distset}.
Since $\gls{dset}^{\mathrm{c}}$ is a polytopic set and inherits compactness from \eqref{eq:distset}, it can be described in terms of its $N_{\mathrm{v}}$ vertices as $\gls{dset}^{\mathrm{c}} = \conv{\{\glsd{dd}_{\glsd{xi}+\glsd{u},\,j}\}_{j=1}^{N_{\mathrm{v}}}}$. Thus, by exploiting \eqref{eq:connection_sysparams_data} and \eqref{eq:consistency_solutions}, we retrieve the corresponding set of consistent system matrices, i.e., 
\begin{equation} \label{eq:setsysmat}
    \mathbb{A} \coloneqq \conv{\left\{\mat{\bm{\Phi}_j & \bm{\Psi}_j}\right\}_{j=1}^{N_{\mathrm{v}}} },
\end{equation}
with the matrix vertices $\mat{\bm{\Phi}_j~~ \bm{\Psi}_j}$.

The set \eqref{eq:setsysmat} of system matrices that are consistent with the available data is the key focus of the data informativity framework \cite{van2023informativity}, set membership identification \cite{milanese1991smi}, and related approaches on robust data-driven controller design \cite{berberich2020robustfb,alanwar2023data}. In contrast to these related works, our focus on the set \eqref{eq:distset_consistency} of consistent disturbance data allows for the construction of data-driven predictors directly based on disturbance data samples, which can be generated by leveraging the available distributional knowledge. Nevertheless, in Section~\ref{sec:properties}, we exploit the connection between consistent disturbance data and consistent system parameters for controller design: the vertices of \eqref{eq:setsysmat} allow for the formulation of a robust constraint on the first predicted step and stabilizing terminal ingredients as common in literature.

Lastly, we remark that the parameterization \eqref{eq:consistency_params} naturally allows for incorporation of potentially available partial model knowledge into the set \eqref{eq:distset_consistency} of consistent disturbance data (and, thus, into controller design) via \eqref{eq:connection_sysparams_data}. This is detailed in Appendix~\ref{app:example}, together with an illustrative example of the findings of this section. In the following section, we detail the design of the proposed DPC scheme based on samples from the set of consistent disturbance data. 

\section{Sampling-based Stochastic DPC} \label{sec:method}
In this section, we elaborate on the design steps of the proposed DPC scheme.
First, we formulate a conceptual OCP that acts as the basis of the proposed DPC scheme. For tractability in the presence of disturbances, we decompose the input into a pre-stabilizing extended state feedback term with gain $\bm{K}$ and a correction term $\gls{v}_k$, i.e., 
\begin{equation} \label{eq:inputdecomp}
    \gls{u}_k = \bm{K} \gls{xi}_k + \gls{v}_k,
\end{equation}
where only the latter is determined by the predictive controller. The gain $\bm{K}$ can be determined purely from data, see \cite{berberich2020data,van2023informativity,de2019formulas}, or Appendix~\ref{app:stabilizingIngredients}. In order to reflect this change of inputs, we construct the data $\glsd{vd}_{T} \coloneqq \{\gls{vd}_{i}\}_{i=1}^{T}$, $\gls{vd}_{i} = \gls{ud}_i - \bm{K} \gls{xid}_i$. Based on Lemma~\ref{lem:extfundamental}, the conceptual OCP associated with the proposed DPC scheme is
\begin{subequations} \label{eq:ocp_original}
		\begin{align}
		& \underset{\bm{\alpha}}{\mathrm{minimize}} ~~~ J_{\gls{Tf}}\left(\gls{u}_{\mathrm{f},k},\,\gls{y}_{\mathrm{f},k}\right) &\label{eq:ocp_original_cost}\\		
		\mathrm{s.t.}~~ &  \vc{\gls{xi}_k \\ \gls{v}_{\mathrm{f},k}\\ \gls{d}_{\mathrm{f},k}\\ \gls{y}_{\mathrm{f},k}} = \mat{\hankel{1}{\glsd{xid}_{T-\gls{Tf}+1}} \\ \hankel{\gls{Tf}}{\glsd{vd}_{T}} \\ \hankel{\gls{Tf}}{\glsd{dd}_{T}} \\ \hankel{\gls{Tf}}{\glsd{yd}_{T}}} \bm{\alpha}, &\label{eq:ocp_original_predictor} \\
		& \Pr{\gls{y}_{l|k} \in \mathbb{Y}} \ge 1-\eps &\hspace{-10mm}\forall \,l \in \mathbb{N}_0^{\gls{Tf}-1}, \label{eq:ocp_original_chancecons} \\
		& \gls{u}_{l|k} = \gls{v}_{l|k} + \bm{K} \gls{xi}_{l|k} \in \mathbb{U} &\hspace{-10mm}\forall \,l \in \mathbb{N}_0^{\gls{Tf}-1},\label{eq:ocp_original_inputcons}\\
		& \gls{xi}_{\gls{Tf}|k} \in \mathbb{X}_{\gls{Tf}},& \label{eq:ocp_original_termcons}
		\end{align}
\end{subequations}
with the vectors $\gls{v}_{\mathrm{f},k} \coloneqq \col{\gls{v}_{0|k},\,\dots,\,\gls{v}_{\gls{Tf}-1|k}}$, $\gls{u}_{\mathrm{f},k} \coloneqq \col{\gls{u}_{0|k},\,\dots,\,\gls{u}_{\gls{Tf}-1|k}}$, $\gls{d}_{\mathrm{f},k} \coloneqq \col{\gls{d}_{k},\,\dots,\,\gls{d}_{k+\gls{Tf}-1}}$, $\gls{y}_{\mathrm{f},k} \coloneqq \col{\gls{y}_{0|k},\,\dots,\,\gls{y}_{\gls{Tf}-1|k}}$, the predicted extended state~$\gls{xi}_{l|k}$ constructed via \eqref{eq:extstate}, and a suitable robust positive invariant (RPI), polytopic terminal constraint set $\mathbb{X}_{\gls{Tf}}$ designed for stability and constructed form data (see Section~\ref{sec:properties} and Appendix~\ref{app:stabilizingIngredients}).
\begin{remark}
    It is implicitly assumed that the data of correction inputs $\glsd{vd}_{T}$ satisfies the PE condition in Lemma~\ref{lem:extfundamental}. If that is not the case for the available data, the input decomposition \eqref{eq:inputdecomp} can directly be considered in the data collection to ensure PE data, e.g., as in \cite{kloppelt2022novel}.
\end{remark}
OCP \eqref{eq:ocp_original} is intractable due to the uncertainty $\bm{w} \coloneqq \left\{\glsd{dd}_T,\,\gls{d}_{\mathrm{f},k}\right\}$, consisting of the unknown disturbance data $\glsd{dd}_{T}$ and future disturbances $\gls{d}_{\mathrm{f},k}$.
We overcome this issue by deterministically approximating OCP~\eqref{eq:ocp_original} using \textit{samples} of the disturbance data $\glsd{dd}_{T}$ and the future disturbances~$\gls{d}_{\mathrm{f},k}$ (from the set of consistent disturbance data \eqref{eq:distset_consistency} and from the disturbance set \eqref{eq:distset}, respectively). Samples are drawn by leveraging Assumption~\ref{assum:distset}; see discussion in Section~\ref{sec:alg}.
\subsection{Data-driven Sample-based Predictions}
First, we derive predictors for future outputs, inputs, and terminal extended state based on the available system data and samples of the uncertainty $\bm{w}$. By Lemma~\ref{lem:extfundamental}, there exists an $\bm{\alpha}(\bm{w})$ that satisfies~\eqref{eq:ocp_original_predictor} with fixed initial condition $\gls{xi}_k$ and sequence of correction inputs $\gls{v}_{\mathrm{f},k}$.
In fact, we can parameterize all solutions for $\bm{\alpha}(\bm{w})$ as
\begin{equation} \label{eq:samplepred_alpha}
    \bm{\alpha}\left(\bm{w}\right) = \mat{\hankel{1}{\glsd{xid}_{\tilde{T}}} \\ \hankel{\gls{Tf}}{\glsd{vd}_{T}} \\ \hankel{\gls{Tf}}{\glsd{dd}_{T}}}^{\dagger} \vc{\gls{xi}_k \\ \gls{v}_{\mathrm{f},k}\\ \gls{d}_{\mathrm{f},k}} + \bm{\Pi}^{0}_{\alpha}\left(\glsd{dd}_{T}\right) \tilde{\bm{\alpha}},
\end{equation}
with $\tilde{T} = T-\gls{Tf}+1$, free variable $\tilde{\bm{\alpha}} \in \mathbb{R}^{\tilde{T}-\glsd{xi}-(\glsd{u}+\glsd{d})\gls{Tf}}$ and some $\bm{\Pi}^0_{\alpha}\left(\glsd{dd}_{T}\right)$ whose columns span the null space of the data matrix $\col{\hankel{1}{\glsd{xid}_{\tilde{T}}},\, \hankel{\gls{Tf}}{\glsd{vd}_{T}},\,\hankel{\gls{Tf}}{\glsd{dd}_{T}}}$.
By denoting the vector of deterministic decision variables as $\gls{var}_k \coloneqq \col{\gls{xi}_k,\,\gls{v}_{\mathrm{f},k},\,\tilde{\bm{\alpha}}} \in \mathbb{R}^{\glsd{var}}$, we can rewrite \eqref{eq:samplepred_alpha} into the form $\bm{\alpha}(\bm{w}) = \bm{M}(\bm{w}) \gls{var} + \bm{m}(\bm{w})$ with
\begin{subequations} \label{eq:samplepred_alpha_reform}
\begin{align}
    \bm{M}(\bm{w}) =& \mat{\left(\mat{\hankel{1}{\glsd{xid}_{\tilde{T}}} \\ \hankel{\gls{Tf}}{\glsd{vd}_{T}} } \bm{\Pi}_{d}(\bm{w}) \right)^{\dagger} & \bm{\Pi}^0_{\alpha}(\glsd{dd}_T)} \\
    \bm{m}(\bm{w}) =& \left(\hankel{\gls{Tf}}{\glsd{dd}_{T}} \bm{\Pi}_{\xi,v} \right)^{\dagger} \gls{d}_{\mathrm{f},k},
\end{align}
\end{subequations}
where $\bm{\Pi}_{\xi,v}$, $\bm{\Pi}_{d}$ are projectors onto the respective null spaces as
\begin{subequations} \label{eq:datamat_helper}
    \begin{align}
        \bm{\Pi}_{\xi,v} &\coloneqq \bm{I}_{\tilde{T}} - \mat{\hankel{1}{\glsd{xid}_{\tilde{T}}} \\ \hankel{\gls{Tf}}{\glsd{vd}_{T}}}^{\dagger} \mat{\hankel{1}{\glsd{xid}_{\tilde{T}}} \\ \hankel{\gls{Tf}}{\glsd{vd}_{T}}},\\
        \bm{\Pi}_{d}(\bm{w}) &\coloneqq \bm{I}_{\tilde{T}} - \hankel{\gls{Tf}}{\glsd{dd}_{T}}^{\dagger} \hankel{\gls{Tf}}{\glsd{dd}_{T}}.
    \end{align}
\end{subequations}

By applying $\bm{\alpha}(\bm{w})$ to \eqref{eq:ocp_original_predictor} and by considering the decomposition \eqref{eq:inputdecomp}, we obtain data-driven predictors for the future outputs, inputs, and terminal extended state depending on the uncertainty $\bm{w}$ as
\begin{subequations} \label{eq:predictors} 
    \begin{align} 
        \gls{y}_{\mathrm{f},k} = \bm{M}_y\left(\bm{w}\right) \gls{var}_k + \bm{m}_y\left(\bm{w}\right), \label{eq:samplepred_output} \\
        \gls{u}_{\mathrm{f},k} = \bm{M}_u\left(\bm{w}\right) \gls{var}_k + \bm{m}_u\left(\bm{w}\right), 
            \label{eq:samplepred_input} \\         
        \gls{xi}_{\gls{Tf}|k} = \bm{M}_{\xi}\left(\bm{w}\right) \gls{var}_k + \bm{m}_{\xi}\left(\bm{w}\right), \label{eq:samplepred_termstate}
    \end{align}
\end{subequations}
with the data-driven predictor parameters
\begin{subequations} 
    \begin{align} 
        \mat{\bm{M}_y ~~ \bm{m}_y} &\coloneqq \hankel{\gls{Tf}}{\glsd{yd}_{T}} \mat{\bm{M}(\bm{w}) ~~ \bm{m}(\bm{w})},\\
        \mat{\bm{M}_u ~~ \bm{m}_u} &\coloneqq \hankel{\gls{Tf}}{\glsd{ud}_{T}} \mat{\bm{M}(\bm{w}) ~~ \bm{m}(\bm{w})},\\
        \mat{\bm{M}_{\xi} ~~ \bm{m}_{\xi}} &\coloneqq \mat{\gls{xid}_{\gls{Tf}+1} \, \cdots \, \gls{xid}_{T+1}} \mat{\bm{M}(\bm{w}) ~~ \bm{m}(\bm{w})}.
    \end{align} 
\end{subequations}

Evidently, the predictors \eqref{eq:predictors} depend on the uncertainty realization $\bm{w}$ and can thus not be applied directly. However, by employing uncertainty samples $\bm{w}^{(i)} \coloneqq \big\{\glsd{dd}^{(i)}_{T},\,\gls{d}^{(i)}_{\mathrm{f},k}\big\}$, $i \in \mathbb{N}_1^{N_{\mathrm{s}}}$, we obtain deterministic predictions $\gls{y}^{(i)}_{\mathrm{f},k}$, $\gls{u}^{(i)}_{\mathrm{f},k}$, and $\gls{xi}^{(i)}_{\gls{Tf}|k}$. Such sample-based predictions allow for reformulation of the constraints~\eqref{eq:ocp_original_chancecons}\textendash\eqref{eq:ocp_original_termcons} and the cost~\eqref{eq:ocp_original_cost} in terms of the deterministic decision variable $\bm{\zeta}_k$, i.e., the measured extended state $\gls{xi}_k$, input sequence $\gls{v}_{\mathrm{f},k}$, and $\tilde{\bm{\alpha}}$.

\begin{remark} \label{rem:pred}
    When consistent disturbance data samples $\glsd{dd}_T^{(i)}$ are used in \eqref{eq:samplepred_alpha}, the predictors \eqref{eq:predictors} are independent of the free variable $\tilde{\bm{\alpha}}$ as the image of $\bm{\Pi}^{0}_{\alpha}\big(\glsd{dd}^{(i)}_{T}\big)$ is entirely contained in the null space of the data matrix in \eqref{eq:ocp_original_predictor}. In other words, the predicted input-output trajectory for the sampled (consistent) uncertainty is uniquely determined via the extended state $\gls{xi}_k$ and sequence of correction inputs $\gls{v}_{\mathrm{f},k}$. Therefore, one can choose $\tilde{\bm{\alpha}} = \bm{0}$ and thus reduce the vector of decision variables to $\gls{var}_k = \col{\gls{xi}_k,\,\gls{v}_{\mathrm{f},k},\,\bm{0}}$. When using consistent disturbance data samples, the predictors \eqref{eq:predictors} coincide with the predictors commonly used in SPC \cite{favoreel1999spc}, and DPC and SPC yield the same predictor \cite{fiedler2021relationship}. Moreover, the predictors \eqref{eq:predictors} are then equivalent to model-based predictors based on the system matrices that correspond to the disturbance data samples $\glsd{dd}_T^{(i)}$ via \eqref{eq:connection_sysparams_data}.
\end{remark}


\subsection{Constraint Sampling \& Reformulation of Cost Function} \label{sec:conssampling}
In order to render the OCP \eqref{eq:ocp_original} tractable, the chance constraint \eqref{eq:outputcons} needs to be reformulated into a deterministic expression. We reformulate the constraints \eqref{eq:ocp_original_chancecons}\textendash\eqref{eq:ocp_original_termcons} in terms of the previously derived data-driven predictors \eqref{eq:predictors}, and deterministically approximate the chance constraint \eqref{eq:ocp_original_chancecons} via sampling of the uncertainty $\bm{w}$. 
Using \eqref{eq:samplepred_output}, we define the set $\mathbb{Y}^{\mathrm{P}}_l$ of (deterministic) decision variables $\gls{var}_k$ for which the predicted output $\gls{y}_{l|k}$, $l\in\mathbb{N}_0^{\gls{Tf}-1}$, satisfies the chance constraint \eqref{eq:ocp_original_chancecons} with probability of at least $1-\eps$, i.e.,
\begin{equation} \label{eq:ocpchancecons_CSS}
    \mathbb{Y}^{\mathrm{P}}_l \coloneqq \left\{ \gls{var}_k ~\left|~ \Pr{\bm{G}_{y} \gls{y}_{l|k} \le \bm{g}_{y}} \ge 1-\eps \right. \right\}.
\end{equation}
Sets of the form \eqref{eq:ocpchancecons_CSS} are commonly referred to as $\eps$-chance constraint sets ($\eps$-CCS) \cite{mammarella2022chance}.
The goal of the constraint sampling is to determine a deterministic inner-approximation $\mathbb{Y}^{\mathrm{S}}_l$ of the $\eps$-CCS \eqref{eq:ocpchancecons_CSS} by using $N_{\mathrm{s}}$ iid samples $\bm{w}^{(i)} = \{\glsd{dd}^{(i)}_{T},\,\gls{d}^{(i)}_{\mathrm{f},k}\}$, $i \in \mathbb{N}_1^{N_{\mathrm{s}}}$, of the uncertainty. 
Given a single uncertainty sample $\bm{w}^{(i)}$ and corresponding predictor \eqref{eq:samplepred_output}, the sampled set corresponding to the output $\eps$-CSS \eqref{eq:ocpchancecons_CSS} reads
\begin{equation} \label{eq:ocpchancecons_sampled}
    \tilde{\mathbb{Y}}^{\mathrm{S}}_l\big(\bm{w}^{(i)}\big) = \big\{ \gls{var}_k ~\left|~ \bm{G}^{(i)}_{y,l} \gls{var}_k \le \bm{g}^{(i)}_{y,l}  \right. \big\},
\end{equation}
with the sample-based constraint parameters
\begin{subequations}
    \begin{align}
        \bm{G}^{(i)}_{y,l} &\coloneqq \bm{G}_{y} \left[\bm{M}_y\big(\bm{w}^{(i)}\big)\right]_{[l\glsd{y}+1:(l+1)\glsd{y}]},\\ \bm{g}^{(i)}_{y,l} &\coloneqq \bm{g}_{y} - \bm{G}_{y} \left[\bm{m}_y\big(\bm{w}^{(i)}\big)\right]_{[l\glsd{y}+1:(l+1)\glsd{y}]}.
    \end{align}
\end{subequations}
Two popular approaches to the sampling-based approximation of $\eps$-CSSs are the direct sampling-based approximation \cite[Lem.~1]{mammarella2022chance} and the probabilistic scaling approach \cite[Th.~1]{mammarella2022chance}.
Both approaches use sampled sets \eqref{eq:ocpchancecons_sampled} to construct inner approximations of the $\eps$-CCS \eqref{eq:ocpchancecons_CSS}. A concise introduction to these sampling approaches, tailored to the setting of this work, is provided in Appendix~\ref{app:samplingprelims}. Using $N_{\mathrm{s}}$ samples $\bm{w}^{(i)}$, we compute a deterministic constraint set $\mathbb{Y}^{\mathrm{S}}_l$ that inner-approximates $\mathbb{Y}^{\mathrm{P}}_l$ with a pre-defined level of confidence $\conf_y$, such that $\Pr{\mathbb{Y}^{\mathrm{S}}_l \subseteq \mathbb{Y}^{\mathrm{P}}_l} \ge 1-\conf_y$. 

The state feedback in \eqref{eq:inputdecomp} introduces uncertainty into the predicted inputs $\gls{u}_{l|k}$. 
In order to accommodate this uncertainty, we approximate the hard input constraints \eqref{eq:ocp_original_inputcons} analogously to the output constraints by employing the predictor \eqref{eq:samplepred_input} based on samples $\bm{w}^{(i)}$, yielding constraint sets $\tilde{\mathbb{U}}_l$ for $l\in\mathbb{N}_0^{\gls{Tf}-1}$, where $\tilde{\mathbb{U}}_0$ is such that the actually applied control input satisfies the hard constraints $\mathbb{U}$.
Similarly, the terminal constraint set \eqref{eq:ocp_original_termcons} is approximated based on the predictor \eqref{eq:samplepred_termstate}, yielding the sampled constraint set $\tilde{\mathbb{X}}_{\gls{Tf}}$.
Finally, by intersecting the sampled constraint sets $\tilde{\mathbb{U}}_l$, $\tilde{\mathbb{Y}}_l$, $\tilde{\mathbb{X}}_{\gls{Tf}}$, $l\in\mathbb{N}_0^{\gls{Tf}-1}$, we collect all constraints on the decision variable into a single set
\begin{equation} \label{eq:constraints_red}
    \mathbb{C} = \left\{ \gls{var}_k ~\left|~ \bm{G}_{\mathrm{OCP}} \gls{var}_k \le \bm{g}_{\mathrm{OCP}}\right. \right\},
\end{equation}
with appropriate $\bm{G}_{\mathrm{OCP}}$, $\bm{g}_{\mathrm{OCP}}$.
If $\gls{var}_k = \col{\gls{xi}_k,\,\gls{v}_{\mathrm{f},k},\,\tilde{\bm{\alpha}}}\in\mathbb{C}$, then constraints \eqref{eq:ocp_original_chancecons}\textendash\eqref{eq:ocp_original_termcons} are satisfied (with confidence).

It remains to deterministically reformulate the expected cost \eqref{eq:cost_expected} in \eqref{eq:ocp_original_cost}.
Using the data-driven predictors \eqref{eq:predictors}, we can express the cost~\eqref{eq:cost_expected} in terms of the decision variable $\gls{var}_k$ and the uncertainty $\bm{w}$ as
\begin{equation} \label{eq:cost_reform}
    J_{\gls{Tf}}(\gls{var}_k) = \norm{\gls{var}_k}^2_{\E{\bm{Q}_{\mathrm{OCP}}(\bm{w})}} + \E{c(\bm{w})},
\end{equation}
with $\tilde{\bm{Q}} \coloneqq \bm{I}_{\gls{Tf}} \otimes \bm{Q}$, $\tilde{\bm{R}} \coloneqq \bm{I}_{\gls{Tf}} \otimes \bm{R}$, and cost parameters
\begin{subequations}
    \begin{align*}
        \bm{Q}_{\mathrm{OCP}}(\bm{w}) &\coloneqq {\bm{M}_y}^{\top} \tilde{\bm{Q}} \bm{M}_y + {\bm{M}_u}^{\top} \tilde{\bm{R}} \bm{M}_u + {\bm{M}_{\xi}}^{\top} \bm{P} \bm{M}_{\xi},\\
        c(\bm{w}) & \coloneqq {\bm{m}_y}^{\top} \tilde{\bm{Q}} \bm{m}_y + {\bm{m}_u}^{\top} \tilde{\bm{R}} \bm{m}_u + {\bm{m}_{\xi}}^{\top} \bm{P} \bm{m}_{\xi}.
    \end{align*}    
\end{subequations}
Note that we discarded terms that are linear in $\gls{d}_{\mathrm{f},k}$ as $\E{\gls{d}_{\mathrm{f},k}} = \bm{0}$ and the disturbances are iid by Assumption~\ref{assum:distset}.
Moreover, since $\E{c(\bm{w})}$ is a constant, it suffices to consider 
\begin{equation} \label{eq:cost_exact}
    J(\gls{var}_k) \coloneqq J_{\gls{Tf}}(\gls{var}_k) - \E{c(\bm{w})} = \norm{\gls{var}_k}^2_{\E{\bm{Q}_{\mathrm{OCP}}(\bm{w})}}.
\end{equation}
as an alternative cost for OCP~\eqref{eq:ocp_original}. As analytical evaluation of the expectation in \eqref{eq:cost_exact} is generally not practical, we employ sample average approximation (SAA). 
Based on $\gls{Nsaa}$ uncertainty samples $\bm{w}^{(i)}$, the SAA cost function that approximates \eqref{eq:cost_exact} is
\begin{equation} \label{eq:cost_sampleaverage}
    \hat{J}\left(\gls{var}_k\right) \coloneqq \norm{\gls{var}_k}^2_{\hat{\bm{Q}}_{\mathrm{OCP}}},
\end{equation}
with weight matrix $\hat{\bm{Q}}_{\mathrm{OCP}} \coloneqq (1/\gls{Nsaa})\sum_{i=1}^{\gls{Nsaa}} \bm{Q}_{\mathrm{OCP}}\big(\bm{w}^{(i)}\big)$.
\begin{remark}
    For the general case of $\gls{var}_k = \col{\gls{xi}_k,\,\gls{v}_{\mathrm{f},k},\,\tilde{\bm{\alpha}}}$ with free $\tilde{\bm{\alpha}}$ (see Remark~\ref{rem:pred}), the cost \eqref{eq:cost_sampleaverage} might be augmented by a regularization term $\mu \norm{\tilde{\bm{\alpha}}}$, $\mu > 0$, to penalize the deviation from the subspace predictor \eqref{eq:samplepred_alpha} with $\tilde{\bm{\alpha}} = \bm{0}$; see discussions in \cite{dorfler2021bridging}.
\end{remark}
\subsection{Control Algorithm} \label{sec:alg}
Based on the previously derived constraints \eqref{eq:constraints_red} and cost function \eqref{eq:cost_sampleaverage}, we define the OCP that is solved at each time-step $k$ as
\begin{align}\label{eq:ocp} 
    \hat{J}^*\left(\gls{xi}_k\right) &=  \min_{\gls{v}_{\mathrm{f},k},\,\tilde{\bm{\alpha}}} \hat{J}\left(\gls{var}_k\right)~~~ 	
    \mathrm{s.t.}~~ \gls{var}_k \in \mathbb{C} \cap \mathbb{C}_{\mathrm{R}}, 
\end{align}
where $\mathbb{C}_{\mathrm{R}}$ guarantees recursive feasibility via a robust constraint on the first predicted step and is constructed in Section~\ref{sec:recfeas}. 
We remark that a solution to \eqref{eq:ocp} can only exist if $\mathbb{C} \,\cap\, \mathbb{C}_{\mathrm{R}}$ is non-empty; sufficient conditions for non-emptiness are discussed in Section~\ref{sec:recfeas}, Remark~\ref{rem:nonemptiness}. 
The implicit control law associated with OCP~\eqref{eq:ocp} reads
\begin{equation} \label{eq:controllaw}
    \bm{\kappa}\left(\gls{xi}_k\right):= \gls{u}^{*}_k = \bm{K} \gls{xi}_k + \gls{v}^{*}_{0|k}
\end{equation}
where $\gls{v}^{*}_{0|k} = \big[\gls{v}^*_{\mathrm{f},k}\big]_{[1:m]}$ is the first input of the optimal input vector $\gls{v}^*_{\mathrm{f},k}$.
Algorithm~\ref{alg:controller} summarizes the overall control scheme.
\begin{algorithm}
\caption{Sampling-based Stochastic DPC for System \eqref{eq:system}} \label{alg:controller}
\begin{algorithmic}[1]
\renewcommand{\algorithmicrequire}{\textbf{Offline Phase:}}
\renewcommand{\algorithmicensure}{\textbf{Online Phase:}}
\REQUIRE
\STATE Retrieve an input-output data trajectory satisfying Assumption~\ref{assum:trajData}.
\STATE Compute the set of consistent disturbance data \eqref{eq:distset_consistency}.
\STATE Determine $\bm{K}$, $\bm{P}$, and $\mathbb{X}_{\gls{Tf}}$satisfying Assumption~\ref{assum:stabilizingIngredients} (e.g., see Appendix~\ref{app:stabilizingIngredients}).
\STATE Compute constraint set $\mathbb{C}$~\eqref{eq:constraints_red} from disturbance samples.
\STATE Compute the first-step constraint $\mathbb{C}_{\mathrm{R}}$ (see Section~\ref{sec:recfeas}). 
\STATE Determine the weight $\hat{\bm{Q}}_{\mathrm{OCP}}$ of the cost function \eqref{eq:cost_sampleaverage}. 
\ENSURE for all $k \ge 0$:
\STATE Construct $\gls{xi}_k$ from most recent $\gls{Tini}$ input-output measurements.
\STATE Solve the OCP \eqref{eq:ocp} to obtain $\gls{v}^{\ast}_{0|k}$.
\STATE Apply the input $\gls{u}_{k} = \gls{v}^{\ast}_{0|k} + \bm{K} \gls{xi}_k$ to the system.
\end{algorithmic}
\end{algorithm}

Since all heavy computations are performed in the offline phase, the online phase of the proposed controller entails only a dense quadratic program for which efficient solvers exist. To further reduce the online computational load, redundant constraints should be removed from the final constraint set $\mathbb{C} \cap \mathbb{C}_{\mathrm{R}}$ in \eqref{eq:ocp}; see \cite{lorenzen2017stochastic,fukuda2020polyhedral} for redundancy removal algorithms.

The data-driven predictors \eqref{eq:predictors} require samples of future disturbances $\gls{d}_{\mathrm{f},k}$ and of (consistent) disturbance data $\glsd{dd}_T$. From Assumption~\ref{assum:distset}, the distribution and the polytopic support set of $\gls{d}_{\mathrm{f},k}$ follow immediately. Similarly for the disturbance data $\glsd{dd}_T$, the distribution over the free variables $\glsd{dd}_{\glsd{xi}+\glsd{u}}$ with support \eqref{eq:distset_consistency} follows from applying the Cartesian product to the individual disturbance distribution $T$ consecutive times, and then exploiting the parameterization \eqref{eq:consistency_solutions}. To generate samples from the polytopic support sets, rejection sampling \cite{martino2010generalized} or Markov Chain Monte Carlo methods \cite{mete2012patternHR} can be employed.
\section{Closed-loop Properties} \label{sec:properties}
We now present properties of the closed-loop system that result from applying the controller~\eqref{eq:controllaw} to system \eqref{eq:system}, or, equivalently, \eqref{eq:system_arx_nonminimal}:
\begin{subequations}
    \begin{align}
        \gls{xi}_{k+1} &= \tilde{\bm{A}} \gls{xi}_{k} + \tilde{\bm{B}} \bm{\kappa}\left(\gls{xi}_{k}\right) + \tilde{\bm{E}}\gls{d}_{k}, \label{eq:system_cl_nonminimal} \\
        \gls{y}_{k} &= \bm{\Phi} \gls{xi}_{k} + \bm{\Psi} \bm{\kappa}\left(\gls{xi}_{k}\right) + \gls{d}_{k}. \label{eq:system_cl}
    \end{align}
\end{subequations}
To guarantee stability, we assume the existence of a suitable stabilizing feedback gain $\bm{K}$ \eqref{eq:inputdecomp}, weight $\bm{P}$ for the terminal cost \eqref{eq:cost_expected}, and RPI terminal constraint $\mathbb{X}_{\gls{Tf}}$ \eqref{eq:ocp_original_termcons}, based on the set \eqref{eq:setsysmat} of consistent system parameters. This is summarized in Assumption~\ref{assum:stabilizingIngredients} and is common in robust and stochastic predictive control \cite{lorenzen2017stochastic,lorenzen2019robust,arcari2023stochastic}. 
Existing literature can be used to determine these ingredients from data by solving semi-definite programs offline, e.g., \cite{van2023informativity,de2019formulas}, and Appendix~\ref{app:stabilizingIngredients}, involving data-driven linear matrix inequalities.

\begin{assumption}[Stabilizing Ingredients] \label{assum:stabilizingIngredients}
    Let $\bm{R}$, $\bm{Q} \succ \bm{0}$ be the weighting matrices from \eqref{eq:cost_expected} and let $\tilde{\mathbb{A}}\coloneqq\conv{\left\{\mat{\tilde{\bm{A}}_j & \tilde{\bm{B}}_j}\right\}_{j=1}^{N_{\mathrm{v}}}}$, $\tilde{\bm{A}}_j \coloneqq \col{\bar{\bm{A}},\, \bm{\Phi}_j}$, $\tilde{\bm{B}}_j \coloneqq \col{\bar{\bm{B}},\, \bm{\Psi}_j}$ be constructed by using \eqref{eq:system_arx_nonminimal_helper} and the vertices $\mat{\bm{\Phi}_j & \bm{\Psi}_j}$, $j \in \mathbb{N}_1^{N_{\mathrm{v}}}$ from \eqref{eq:setsysmat}.
    There exist a gain $\bm{K}$ and weighting matrix $\bm{P} = \bm{P}^{\top}  \succ \bm{0}$ such that for all $j \in \mathbb{N}_1^{N_{\mathrm{v}}}$
    \begin{itemize}
        \item[(a)] $\tilde{\bm{A}}_{\mathrm{cl},j} \coloneqq \tilde{\bm{A}}_j + \tilde{\bm{B}}_j \bm{K}$ is Schur stable, and
        \item[(b)] $\tilde{\bm{A}}^{\top}_{\mathrm{cl},j} \bm{P} \tilde{\bm{A}}_{\mathrm{cl},j} - \bm{P} + \bm{K}^{\top} \bm{R} \bm{K} + \tilde{\bm{A}}^{\top}_{\mathrm{cl},j} \tilde{\bm{E}} \bm{Q} \tilde{\bm{E}}^{\top} \tilde{\bm{A}}_{\mathrm{cl},j} \prec \bm{0} $,
    \end{itemize}
with $\tilde{\bm{E}} = \col{\bm{0},\,\bm{I}_{\glsd{y}}}$. Furthermore, the polytopic terminal set $\mathbb{X}_{\gls{Tf}}$ for system \eqref{eq:system} is RPI under the control law $\gls{u}_k = \bm{K} \gls{xi}_k$, and the constraints \eqref{eq:constraints} are satisfied $\forall\, \gls{xi}_k \in \mathbb{X}_{\gls{Tf}}$. 
\end{assumption}
Assumption~\ref{assum:stabilizingIngredients} is essentially an assumption about the informativity of the available data, akin to data informativity for quadratic stabilization \cite{van2023informativity}, due to the connection of the available data and the corresponding set of consistent system matrices \eqref{eq:setsysmat}, see Section~\ref{sec:distset_consistency}. Feasibility of the methods in Appendix~\ref{app:stabilizingIngredients} is sufficient for Assumption~\ref{assum:stabilizingIngredients}.

\subsection{Recursive Feasibility} \label{sec:recfeas}
In order to render the OCP \eqref{eq:ocp} recursively feasible, we construct an additional constraint $\mathbb{C}_{\mathrm{R}}$ on the first predicted step \cite{lorenzen2017stochastic}.
Let $\mathbb{C}_{\gls{Tf}}$ denote the set of feasible initial states and first inputs, i.e.,
\begin{equation} \label{eq:constraints_red_projected}
    \mathbb{C}_{\gls{Tf}} \coloneqq \left\{ \col{\gls{xi}_k,\,\gls{v}_{0|k}} \left|~ 
         \exists \, \gls{v}_{1|k},\, \dots,\, \gls{v}_{\gls{Tf}-1|k}, \tilde{\bm{\alpha}}:\,
         \gls{var}_k \in \mathbb{C}\right. \right\},
\end{equation}
computed by projection of \eqref{eq:constraints_red}. Based on the vertices $\tilde{\bm{A}}_{\mathrm{cl},j}$, $j \in \mathbb{N}_1^{N_{\mathrm{v}}}$ (see Assumption~\ref{assum:stabilizingIngredients}) and the disturbance bound \eqref{eq:distset}, we determine a robust control invariant set $\gls{roa}$ for system \eqref{eq:system_arx_nonminimal} with $ \col{\gls{xi}_k,\,\gls{v}_{0|k}} \in \mathbb{C}_{\gls{Tf}}$ \cite[Sec.~5.3]{blanchini2015set}. 
At last, we construct the first-step constraint set
\begin{equation} \label{eq:constraints_firststep}
    \mathbb{C}_{\mathrm{R}} \coloneqq \left\{ \gls{var}_k \left| \hspace{-1mm} \begin{array}{l} 
	\forall \,\gls{d} \in \gls{dset},\, j \in \mathbb{N}_1^{N_\mathrm{v}}: \\
    \tilde{\bm{A}}_{\mathrm{cl},j} \gls{xi}_k +\tilde{\bm{B}}_{j} \gls{v}_{0|k} + \tilde{\bm{E}} \gls{d} \in \gls{roa}
	\end{array} \hspace{-2mm}\right. \right\}.
\end{equation}
Since the constraint set $\mathbb{C} \cap \mathbb{C}_{\mathrm{R}}$ is RPI for the closed-loop dynamics~\eqref{eq:system_cl_nonminimal}, OCP~\eqref{eq:ocp} is recursively feasible.
\begin{theorem}[Recursive Feasibility \& Constraint Satisfaction] \label{th:recfeas}
   Consider the set of all feasible input sequences for given $\gls{xi}_{k}$, i.e.,
    \begin{equation} \label{eq:feasset}
        \mathbb{F}\left(\gls{xi}_{k}\right) = \left\{ \gls{v}_{\mathrm{f},k} ~\left|~ \exists \tilde{\bm{\alpha}}:~\gls{var}_k \in \mathbb{C} \cap \mathbb{C}_{\mathrm{R}} \right.\right\}.
    \end{equation} 
    Under the control law \eqref{eq:controllaw},
    $\mathbb{F}\left(\gls{xi}_{k}\right) \neq \emptyset \implies \mathbb{F}\left(\gls{xi}_{k+1}\right) \neq \emptyset$ holds for every realization of $\gls{d}_{k} \in \gls{dset}$. Moreover, for $\gls{xi}_0 \in \gls{roa}$, the closed-loop system \eqref{eq:system_cl} satisfies the output chance constraints \eqref{eq:outputcons} with confidence $\conf_y$ and the hard input constraints \eqref{eq:inputcons} for all $k \ge 0$.
\end{theorem}
\begin{proof}
      The proof follows the same arguments as in \cite[Prop.~9~\&~10]{lorenzen2017stochastic}; for details, see Appendix~\ref{app:recfeasproof}.
\end{proof}
\begin{remark} \label{rem:nonemptiness}
    Crucially, the construction of $\mathbb{C}_{\mathrm{R}}$ \eqref{eq:constraints_firststep} and feasibility of the OCP~\eqref{eq:ocp} rely on non-emptiness of the constraint set $\mathbb{C}$ \eqref{eq:constraints_red}. Given an RPI terminal constraint set $\mathbb{X}_{\gls{Tf}}$ via Assumption~\ref{assum:stabilizingIngredients}, $\mathbb{C}$ is non-empty by design: Consider the direct sampling approach \cite[Lem.~1]{mammarella2022chance} for the $\eps$-CSS approximation in Section~\ref{sec:conssampling}. For $N_{\mathrm{s}}~\rightarrow~\infty$, this corresponds to robust constraint handling. Due to the RPI property of $\mathbb{X}_{\gls{Tf}}$, the RPI terminal set $\mathbb{X}_{\gls{Tf}}$ must be contained in the set of feasible initial states $\mathbb{C}_{\xi} \subset \mathbb{C}$, which renders $\mathbb{C}$ non-empty even for this robust case. As finite $N_{\mathrm{s}}$ only soften the constraint handling, the argument still holds. A similar argument can be made if the probabilistic scaling approach \cite[Th.~1]{mammarella2022chance} is used for the $\eps$-CSS approximation with appropriate choice of approximating sets. 
\end{remark}

\subsection{Robust Asymptotic Stability in Expectation} \label{sec:stability}
We now analyze convergence properties of the closed-loop system.
Literature on stochastic predictive control often provides mean-square stability guarantees via average asymptotic cost bounds, e.g., \cite{pan2023data,mammarella2018offline,arcari2023stochastic}. In contrast, we consider a stronger notion of stability in this work, namely robust asymptotic stability in expectation (RASiE).
\begin{definition}[RASiE~\cite{mcallister2022nonlinear}]
    Let $\gls{roa}$ be a closed RPI set for system \eqref{eq:system_cl_nonminimal} with $\bm{0} \in \gls{roa}$, and let $\gls{xi}_k$ denote the solution to \eqref{eq:system_cl_nonminimal} at time $k \in \mathbb{N}_0$ for given initial condition $\gls{xi}_0$ and disturbance trajectory $\left\{\gls{d}_0,\,\dots,\,\gls{d}_k\right\}$. The origin of system \eqref{eq:system_cl_nonminimal} is robustly asymptotically stable in expectation on $\gls{roa}$ for a given distribution of $\bm{d}$ (Assumption~\ref{assum:distset}) and its associated covariance $\bm{\Sigma}\coloneqq \E{\gls{d}\gls{d}^{\top}}$ if there exist functions $\beta\in\mathscr{KL}$ and $\rho\in\mathscr{K}$ such that
    \begin{equation} \label{eq:rasie_def}
    	\E{\norm{\gls{xi}_k}} \le \beta\left(\norm{\gls{xi}_0},\,k\right) + \rho\left(\trace{\bm{\Sigma}}\right)
    \end{equation}
    for all $k \in \mathbb{N}_0$ and $\gls{xi}_0 \in \gls{roa}$.
\end{definition}
RASiE provides a uniform bound on the expected value of the norm of the closed-loop state $\gls{xi}_k$ depending on the initial condition $\gls{xi}_0$ and the disturbance covariance $\bm{\Sigma}$, and it ensures that the effect of $\gls{xi}_0$ on this bound asymptotically decays towards zero \cite{mcallister2022nonlinear}. 
In contrast to input-to-state stability (ISS) \cite{jiang2001input}, RASiE considers the expected value of the norm of $\gls{xi}_k$ and the disturbance covariance $\bm{\Sigma}$. 
RASiE can be established with the help of a stochastic ISS Lyapunov function \cite[Def.~3 ~\&~Prop.~13]{mcallister2022nonlinear}. For the construction of such Lyapunov functions, the expected cost \eqref{eq:cost_exact} is commonly employed. However, we only have access to the SAA cost \eqref{eq:cost_sampleaverage}. The following lemma provides a bound on the difference between \eqref{eq:cost_exact} and \eqref{eq:cost_sampleaverage} depending on the number \gls{Nsaa} of SAA samples.
\begin{lemma}[SAA cost bound] \label{lem:saabound}
    Let $\delta_{\mathrm{s}} \in [0,\,1)$, $\lambda_Q \ge \max_{\bm{w}} \evmax{\bm{Q}_{\mathrm{OCP}}(\bm{w})}$, and $\norm{\gls{xi}_k}^2_{\bm{P}_{\mathrm{c}}} \ge \norm{\gls{var}^*_k}^2$ for all $\gls{xi}_k \in \gls{roa}$ with $\gls{var}^*_k \coloneqq \col{\gls{xi}_k,\,\gls{v}^*_{\mathrm{f},k},\,\bm{0}}$ and feasible $\gls{v}^*_{\mathrm{f},k}$ from \eqref{eq:ocp}. Then,
    \begin{equation} \label{eq:saacostbound}
        \Pr{\forall \gls{xi}_k \in \gls{roa}:~~ \big\lvert \hat{J}^*(\gls{xi}_k) - J(\gls{var}^*_k)\big\rvert \le\tau \norm{\gls{xi}_k}^2_{\bm{P}_{\mathrm{c}}}} \ge \delta_{\mathrm{s}}
    \end{equation}
    holds for $\tau\left(\gls{Nsaa},\,\delta_{\mathrm{s}}\right) \coloneqq \sqrt{2/( \gls{Nsaa})\ln{(2\glsd{var}/(1-\delta_{\mathrm{s}}}))}\lambda_Q$.
\end{lemma}
\begin{proof}
    The result is obtained by applying the matrix Hoeffding inequality \cite[Thm.~6.15]{wainwright2019high} to $\hat{\bm{Q}}_{\mathrm{OCP}} - \E{\bm{Q}_{\mathrm{OCP}}(\bm{w})}$ from the cost functions \eqref{eq:cost_exact} and \eqref{eq:cost_sampleaverage}; details are reported in Appendix~\ref{app:saabound}.
\end{proof}
\begin{remark}
    Note that the uncertainty in $\bm{Q}_{\mathrm{OCP}}(\bm{w})$ is bounded via \eqref{eq:distset_consistency}, thus an upper bound $\lambda_Q$ exists and can be computed offline. The weight $\bm{P}_{\mathrm{c}}$ can be found by using the vertices of the set of feasible initial extended states $\gls{roa}$. Although the cost bound \eqref{eq:saacostbound} only holds with confidence $\delta_{\mathrm{s}}$, the number $\gls{Nsaa}$ of SAA samples can be chosen appropriately large to guarantee a tight bound with high confidence. Finally, we remark that the probability in \eqref{eq:saacostbound} holds uniformly for all $\gls{xi}_k \in \gls{roa}$, which is crucial for the subsequent stability analysis.
\end{remark}
To show descent in the Lyapunov function, stability proofs for predictive controllers commonly rely on the availability of a feasible candidate solution $\tilde{\gls{v}}_{\mathrm{f},k+1} \in \mathbb{F}\left(\gls{xi}_{k+1}\right)$. For sampling-based predictive controllers, guaranteeing feasibility of a candidate solution at all time-steps is in general not possible \cite{lorenzen2017stochastic}. Nonetheless, stability of the closed-loop system can be guaranteed if the probability of infeasibility of the candidate solution is sufficiently low \cite{lorenzen2017stochastic}. We define the candidate solution for the OCP~\eqref{eq:ocp} as follows (cf. \cite{mammarella2018offline}).
\begin{definition}[Candidate Solution] \label{def:candsol}
    Given the OCP \eqref{eq:ocp} and a feasible solution $\gls{v}^*_{\mathrm{f},k}$ at time-step $k$, the candidate solution for time-step $k+1$ is defined as $\tilde{\gls{v}}_{\mathrm{f},k+1} \coloneqq\col{\tilde{\gls{v}}_{0|k+1},\,\dots,\,\tilde{\gls{v}}_{\gls{Tf}-1|k+1}}$ with
    \begin{equation} \label{eq:candsol}
        \tilde{\gls{v}}_{l|k+1} \coloneqq \left\{ \begin{array}{lc}
            \gls{v}^*_{l+1|k} + \bm{K}\tilde{\bm{A}}_{\mathrm{cl}}^l\tilde{\bm{E}}\gls{d}_k, & l \in \mathbb{N}_0^{\gls{Tf}-2}, \\
             \bm{K}\tilde{\bm{A}}_{\mathrm{cl}}^l\tilde{\bm{E}}\gls{d}_k & l = \gls{Tf}-1.
        \end{array}\right.
    \end{equation}
\end{definition}
We now present our main result, namely RASiE of the origin of the closed-loop system~\eqref{eq:system_cl_nonminimal} under the proposed control law~\eqref{eq:controllaw}. 
In contrast to \cite{mammarella2018offline} where availability of the exact expected cost is assumed, we explicitly consider the SAA cost~\eqref{eq:cost_sampleaverage} via Lemma~\ref{lem:saabound}, yielding RASiE with high confidence. 

First, note that if the candidate solution is infeasible, the expected cost increase can be bounded by making use of an lower and upper bound to the optimal SAA cost $\hat{J}^*(\gls{xi}_k)$ from \eqref{eq:ocp}, cf. \cite{mammarella2018offline}. That is,
\begin{equation} \label{eq:costbound}
    \norm{\gls{xi}_k}_{\bm{P}_{\mathrm{l}}}^2 \le \hat{J}^*(\gls{xi}_k) \le \norm{\gls{xi}_k}_{\bm{P}_{\mathrm{u}}}^2~~~\forall \gls{xi}_k \in \gls{roa}
\end{equation}
with suitable matrices $\bm{P}_{\mathrm{l}},\,\bm{P}_{\mathrm{u}}\succ \bm{0}$.
The matrices ${\bm{P}}_{\mathrm{l}}$ and ${\bm{P}}_{\mathrm{u}}$ can be determined using the unconstrained infinite horizon cost and the set of feasible initial extended states $\gls{roa}$ \cite{lorenzen2017stochastic}.
Second, since system \eqref{eq:system_arx_nonminimal} is detectable, there exists an IOSS Lyapunov function $W\left(\cdot\right)$ satisfying
\begin{align}
    & W\left(\gls{xi}_{k+1}\right) -  W\left(\gls{xi}_k\right) \coloneqq \norm{\gls{xi}_{k+1}}^2_{\bm{P}_W} - \norm{\gls{xi}_k}^2_{\bm{P}_W}\notag \\
    & \le -\frac{1}{2}\norm{\gls{xi}_k}^2 + c_u \norm{\gls{u}_k}^2 + c_y \norm{\gls{y}_k}^2 + c_d \norm{\gls{d}_k}^2, \label{eq:IOSSdescent}
\end{align}
with suitable parameters $\bm{P}_W \succ \bm{0}$, $c_u$, $c_y$, $c_d > 0$ that can be determined using $\tilde{\mathbb{A}}$ from Assumption~\ref{assum:stabilizingIngredients}, cf. \cite{bongard2022robust,berberich2021design,cai2008input}.

\begin{theorem}[RASiE of the Closed-loop] \label{th:rasie}
    Let $\eps_{\mathrm{f}} \in [0,\,1)$ be an upper bound of the probability that the candidate solution \eqref{eq:candsol} is not feasible. With confidence $\delta_{\mathrm{s}}$ from \eqref{eq:saacostbound} and $\tau\left(\gls{Nsaa},\,\delta_{\mathrm{s}}\right)$ from Lemma~\ref{lem:saabound}, the origin of system \eqref{eq:system_cl_nonminimal} is robustly asymptotically stable in expectation on $\gls{roa}$ if, for all $\tilde{\bm{A}}_j$, $\tilde{\bm{B}}_j$, $j \in \mathbb{N}_1^{N_{\mathrm{v}}}$ from Assumption~\ref{assum:stabilizingIngredients},
    \begin{equation}
        \bm{Q}^{\mathrm{S}}_j \coloneqq \mat{\frac{c_{\mathrm{S}}}{2}\bm{I}_{\glsd{xi}} & \bm{0} \\ \bm{0} & \bm{R} - c_{\mathrm{S}} c_u\bm{I}_{\glsd{u}}} - \frac{\eps_{\mathrm{f}}}{1-\eps_{\mathrm{f}}}\bm{T}^{a}_{j} - \tau\bm{T}^{b}_{j} \succ \bm{0}, \label{eq:stabcond}
    \end{equation}
    where $c_{\mathrm{S}} \coloneqq \min(\evmin{\bm{Q}},\,\evmin{\bm{R}})/\max(c_u,\,c_y)$, $\bm{T}^{a}_{j} \coloneqq  \bm{T}_j(\bm{P}_{\mathrm{u}},\bm{P}_{\mathrm{l}})$, and $\bm{T}^{b}_{j} \coloneqq T_j(\bm{P}_{\mathrm{c}},-\bm{P}_{\mathrm{c}})$, with
    \begin{equation}
    \bm{T}_j(\bm{P}_{\mathrm{u}},\bm{P}_{\mathrm{l}}) \coloneqq \mat{
        \tilde{\bm{A}}_j^{\top}\bm{P}_{\mathrm{u}}\tilde{\bm{A}}_j - \bm{P}_{\mathrm{l}} & \tilde{\bm{A}}_j^{\top}\bm{P}_{\mathrm{u}}\tilde{\bm{B}}_j\\
        \tilde{\bm{B}}_j^{\top}\bm{P}_{\mathrm{u}}\tilde{\bm{A}}_j & \tilde{\bm{B}}_j^{\top}\bm{P}_{\mathrm{u}}\tilde{\bm{B}}_j
        }.
    \end{equation}
\end{theorem}
\begin{proof}
    The claim follows from \cite[Prop.~13]{mcallister2022nonlinear} with the stochastic ISS Lyapunov function $V\left(\gls{xi}_k\right) \coloneqq \hat{J}^*\left(\gls{xi}_k\right) + \left(1-\eps_{\mathrm{f}}\right)c_{\mathrm{S}} W\left(\gls{xi}_k\right)$. Descent of $V\left(\gls{xi}_k\right)$ is shown by utilizing \eqref{eq:saacostbound}, \eqref{eq:costbound}, and \eqref{eq:IOSSdescent}. Specifically, the bound on the expected increase of cost $\hat{J}^*\left(\gls{xi}_k\right)$ is derived depending on the probability of infeasibility of the candidate solution \eqref{eq:candsol} (as in \cite[Appendix~B]{mammarella2018offline}) and the SAA confidence bound~\eqref{eq:saacostbound}.
    See Appendix~\ref{app:stabproof} for the detailed proof.
\end{proof}

\begin{remark}
    For the trivial case $\eps_{\mathrm{f}} = 0$ (i.e., the candidate solution \eqref{eq:candsol} is always feasible) and for a suitable choice of $\delta_{\mathrm{s}}$, $\gls{Nsaa}$, and the cost parameters $\bm{Q}$, $\bm{R}$, the stability criterion \eqref{eq:stabcond} holds by design.
    However, a priori verification of Theorem~\ref{th:rasie} requires knowledge of $\eps_{\text{f}}$. 
    Monte Carlo methods can be used to bound $\eps_{\text{f}}$ up to a desired level of confidence \cite{homem2014monte}: using uncertainty samples, one can construct corresponding candidate solutions \eqref{eq:candsol} for every vertex of $\mathbb{C}_{\xi}^{\infty}$, and then test OCP \eqref{eq:ocp} for feasibility.
\end{remark}

\begin{remark}
    The guarantee of RASiE in Theorem~\ref{th:rasie} is based on a probabilistic cost bound (Lemma~\ref{lem:saabound}). Alternatively, the Wasserstein distance between the true distribution and the empirical distribution from the SAA samples can be leveraged to guarantee a distributionally robust version of RASiE~\cite{mcallister2023inherent}. In contrast, this leads to performance bounds as in \eqref{eq:rasie_def} that additionally depend on said Wasserstein distance, which is not the case for the proposed result.
\end{remark}

We close this section by remarking that a similar bound as in \eqref{eq:rasie_def} can be given for the output $\gls{y}_k$ using $\gls{y}_k = \tilde{\bm{E}}^{\top}\gls{xi}_{k+1}$.

\section{Numerical Evaluation} \label{sec:eval}
This section evaluates the proposed data-driven sample-based predictors and control algorithm in simulation. Consider the following linearized model of  a DC-DC converter~\cite{lazar2008input} with $\gls{Tini} = 1$ and
\begin{equation}
\bm{\Phi} = \mat{4.697 & 1 & 0.073 \\ 
                 0.083 & -0.060 & 0.997},~~\bm{\Psi} = \bm{0}.
\end{equation}
The system is subject to polytopic input and output constraints with $\norminf{\gls{u}_k}\le0.2$, $\norminf{\gls{y}_k}\le 3$, and subject to uniformly distributed disturbances with $\norminf{\left[\gls{d}_k\right]_1}\le 0.1$ and $\norminf{\left[\gls{d}_k\right]_2}\le 0.05$.
For the data collection (see Assumption~\ref{assum:trajData}), we apply random admissible inputs and record an input-output trajectory of length $T=55$. For the cost function~\eqref{eq:cost_expected}, we choose a prediction horizon of $\gls{Tf} = 6$ and the input and output weighting matrices $\bm{R} = 1,~~\bm{Q} = \mat{\col{1,\,0}~~\col{0,\,100}}$.
The feedback gain $\bm{K}$, terminal weight $\bm{P}$, and terminal set $\mathbb{X}_{\gls{Tf}}$ are computed as in Appendix~\ref{app:stabilizingIngredients} such that Assumption~\ref{assum:stabilizingIngredients} is satisfied.

We compare 4 different cases regarding the disturbance data uncertainty: 1) using \textit{consistent} samples (the proposed method), 2) using \textit{inconsistent} samples (i.e., neglecting the consistency constraint~\eqref{eq:consistency}), 3) using an \textit{estimate} (akin to \cite{pan2021stochastic}), and 4) using the \textit{exact} disturbance data. As the disturbance is uniformly distributed, maximum likelihood estimation does not admit a unique solution. Thus, we employ the Chebychev center of \eqref{eq:distset_consistency} to obtain the disturbance data estimate.

\subsubsection{Open-loop Results}
We first evaluate the open-loop prediction accuracy of the data-driven output predictor~\eqref{eq:samplepred_output}. First, we set $\gls{xi}_k = \bm{0}$ and generate $100$ random admissible input sequences $\gls{v}_{\mathrm{f},k}$. For every input sequence, we draw $1\,000$ samples of $\glsd{dd}^{(i)}$, and compute the corresponding output prediction $\gls{y}_{\mathrm{f},k}^{(i)}$ using~\eqref{eq:samplepred_output} with $\gls{d}_{\mathrm{f},k} = \bm{0}$. Then, we compute the root-mean-square error (RMSE) between the (sampled) output predictions and the exact output trajectory that results from applying the given input sequence. An example scenario is depicted in Fig.~\ref{fig:pred}. For this simulation, predictors based on consistent disturbance data samples emit a mean reduction in RMSE of $47.19\,\%$ compared to using inconsistent samples, and $33.05\,\%$ compared to using the disturbance data estimate.

\begin{figure}
    \centering
    \includegraphics[clip, trim=3mm 22mm 2mm 1mm]{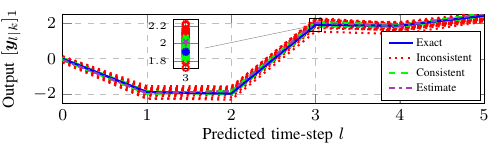} 
    \caption{Comparison of (sampled) open-loop output predictions.}
    \label{fig:pred}
\end{figure}

\subsubsection{Closed-loop Results}
We now evaluate the performance of the proposed controller for the different cases detailed above. With $\bm{\xi}_0 \coloneqq \col{0,\,0,\,2.8}$, the control goal is to stabilize the origin of the system while satisfying constraints~\eqref{eq:constraints}. In a Monte Carlo simulation of $100$ runs, the controller is applied for $30$ time-steps, where at each time-step a newly generated disturbance realization affects the system. The simulations are carried out in MATLAB using the \texttt{quadprog} solver. For the constraint sampling, we employ the direct approximation \cite[Lem.~1]{mammarella2022chance} with the risk parameter $\eps = 0.05$ and confidence $\conf_y = 10^{-4}$.

Fig.~\ref{fig:traj} shows trajectories from 50 exemplary runs of the controller based on consistent disturbance data samples. The probabilistic constraint tightening of the proposed scheme allows the system to operate close to the constraint boundary, leading to fast convergence. No constraint violations occur in any run due to the conservatism of both the robust first-step constraint and the sampling-based $\eps$-CSS approximation. Table~\ref{tab:cost} compares relative increase in trajectory cost $J_{\mathrm{tot}}= \textstyle \sum_{k=0}^{29} \big(\gls{y}_{k}^{\top} \bm{Q}  \gls{y}_{k} + \gls{u}_{k}^{\top}\bm{R}\gls{u}_{k} \big)$ compared to the exact data case. We observe that using consistent rather than inconsistent disturbance data samples results in significantly less cost increase. Using the disturbance data estimate can emit even lower costs than the exact data case, but also results in loss of closed-loop guarantees from Section~\ref{sec:properties}, and is thus not suitable for addressing Problem~\ref{problem}.

\begin{figure}
    \centering
    \includegraphics[clip, trim=3mm 1mm 2mm 26mm]{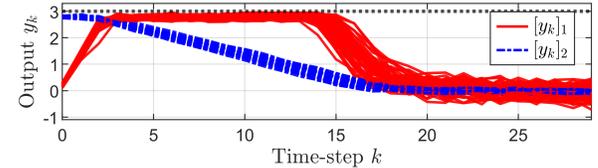}
    \caption{Trajectories of 50 exemplary runs under the proposed controller subject to disturbances. Constraints are shown in dotted black lines.}
    \label{fig:traj}
\end{figure}

\begin{table}
    \centering
    
    \caption{Increase of trajectory cost $J_{\mathrm{tot}}$ relative to exact case }
    \begin{tabular}{lrrrr} \hline
         & Min~~ & Mean~ & Max~~ & Std Dev \\ \hline
        Consistent: & $0.66\,\%$ & $0.81\,\%$ & $2.64\,\%$ & $0.23\,\%$\\
        Inconsistent:  & $8.24\,\%$ & $11.02\,\%$ & $14.60\,\%$ & $1.39\,\%$\\ 
        Estimate:  & $-0.22\,\%$ & $-0.09\,\%$ & $0.96\,\%$ & $0.20\,\%$\\ \hline
    \end{tabular}
    \label{tab:cost}
    
\end{table}

\section{Conclusion} \label{sec:conclusion}
In this paper, we proposed a stochastic output-feedback DPC scheme for the data-driven control of LTI systems subject to bounded additive disturbances. Opposed to related approaches that rely on exact disturbance measurements, we leverage a novel parameterization of the unknown consistent disturbance data for sampling-based approximation of chance constraints. This parameterization implicitly translates the distribution over disturbance data into a distribution over model parameters. Thus, our approach can be seen as the direct counterpart to an indirect approach of first identifying a distribution over models from the data, and then employing stochastic MPC \cite{teutsch2024adaptive}. Closed-loop constraint satisfaction and RASiE hold with predefined confidence under standard assumptions. A numerical example demonstrates that the use of consistent disturbance data samples allows for significant improvement in prediction and control performance.

\appendix


\subsection{On Verifiability of Assumption~\ref{assum:trajData}} \label{app:verify}

Lemma~\ref{lem:extfundamental} relies on a PE condition involving both input and disturbance data, see Assumption~\ref{assum:trajData}. Since the disturbance data are not available in our setting, one may ask how Assumption~\ref{assum:trajData} can be verified. The following lemma connects the PE condition with the rank of an input-output data matrix.
\begin{lemma}[Rank of Data Matrices]\label{lem:rank}
    Consider a control-lable LTI system $\Sigma$ of the form \eqref{eq:system_minimal} and the data trajectories $\glsd{ud}_T$, $\glsd{dd}_T$, $\glsd{yd}_T$, and $\glsd{xid}_{T}$ where $T \ge L \coloneqq \gls{Tini} + \gls{Tf}$, $\gls{Tf} \in \mathbb{N}$, and $\gls{Tini} \ge \lag{\Sigma}$ holds. If the trajectory of generalized inputs $\tilde{U}_T \coloneqq\left\{\col{\gls{ud}_i,\, \gls{dd}_i}\right\}_{i=-\gls{Tini}+1}^{T}$ is PE of order $n + L$, then
    \begin{equation} \label{eq:PErank}
        \rank{\mat{\bm{H}_x \\ \hankel{L}{\tilde{U}_T}}} = \glsd{x} + (\glsd{u}+\glsd{y})L
    \end{equation}
    holds, with $\bm{H}_x = \mat{\gls{xd}_1 ~\cdots~\gls{xd}_{T-\gls{Tf}+1}}$. Furthermore, we have
    \begin{equation} \label{eq:rankeq}
        \rank{\mat{\hankel{1}{\glsd{xid}_{T-\gls{Tf}+1}} \\ \hankel{\gls{Tf}}{\glsd{ud}_{T}} \\ \hankel{\gls{Tf}}{\glsd{yd}_{T}}}} = \rank{\mat{\bm{H}_x \\ \hankel{L}{\tilde{U}_T}}}.
    \end{equation}
\end{lemma}
\begin{proof}
    As system~\eqref{eq:system_minimal} is controllable, \eqref{eq:PErank} follows directly from the PE condition on the generalized input $\col{\gls{u}_k,\, \gls{d}_k}$ \cite[Cor.~2]{willems2005note}.    
    Now, consider the extended observability matrix $\mathcal{O}_{L} \coloneqq \col{\bm{C}, \bm{C} \bm{A}, \ldots, \bm{C} \bm{A}^{L-1}}$ and the Toeplitz matrices $\mathcal{T}_{L}^{u} \coloneqq \mathcal{T}_{L}(\bm{A},\,\bm{B},\,\bm{C},\,\bm{D})$ and $\mathcal{T}_{L}^{d} \coloneqq \mathcal{T}_{L}(\bm{A},\,\bm{E},\,\bm{C},\,\bm{I}_p)$ with
    \begin{equation}
        \mathcal{T}_{L}(\bm{A},\,\bm{B},\,\bm{C},\,\bm{D}) \coloneqq \mat{\bm{D} & \bm{0} & \cdots & \bm{0} \\ \bm{C}\bm{B} & \bm{D} & \ddots & \vdots \\ \vdots & \ddots & \ddots & \bm{0} \\ \bm{C}\bm{A}^{L-2}\bm{B} & \cdots & \bm{C}\bm{B} & \bm{D}}.
    \end{equation}
    With suitable permutation matrices $\bm{\Pi}_{uy}$ and $\bm{\Pi}_{\tilde{u}}$, we have
    \begin{equation} \label{eq:rankeq_helper}
        \bm{\Pi}_{uy}\mat{\hankel{1}{\glsd{xid}_{\tilde{T}}} \\ \hankel{\gls{Tf}}{\glsd{ud}_{T}} \\ \hankel{\gls{Tf}}{\glsd{yd}_{T}}} = \mat{\bm{0} & \bm{I}_{\glsd{u}\gls{Tf}} & \bm{0} \\ \bm{\mathcal{O}}_L & \bm{\mathcal{T}}^u_L & \bm{\mathcal{T}}^d_L} \bm{\Pi}_{\tilde{u}} \mat{\bm{H}_x \\ \hankel{L}{\tilde{U}_T}},
    \end{equation}
    where $\tilde{T} \coloneqq T-\gls{Tf}+1$, cf. \cite[Sec.~2]{de2019formulas}. Note that the permutation matrices $\bm{\Pi}_{uy}$ and $\bm{\Pi}_{\tilde{u}}$ are square and invertible by design. Since $L > \gls{Tini}$, $\bm{\mathcal{O}}_L$ has full column rank. Since $\bm{\mathcal{T}}^d_L$ is a lower triangular matrix with ones on the diagonal, also $\bm{\mathcal{T}}^d_L$ has full column rank. Thus, \eqref{eq:rankeq} follows from \eqref{eq:rankeq_helper}.
\end{proof}

Condition \eqref{eq:PErank} is necessary for Lemma~\ref{lem:extfundamental} (cf. \cite[Lem.~2]{de2019formulas}), and via Lemma~\ref{lem:rank}, the PE condition on the data in Assumption~\ref{assum:trajData} guarantees that \eqref{eq:PErank} holds. In order to verify whether \eqref{eq:PErank} is satisfied without access to disturbance data, one can evaluate the left-hand side of \eqref{eq:rankeq}. However, note that the PE condition in Assumption~\ref{assum:trajData} is only sufficient and not necessary for \eqref{eq:PErank}; we refer to \cite{van2021beyond} for a further discussion on this topic.

\subsection{Set of Consistent Disturbance Data: Incorporating Prior Model Knowledge and Illustrative Example} \label{app:example}
Here, we discuss how to potentially available prior model knowledge (e.g., knowledge on the structure and entries of $\bm{\Phi}$, $\bm{\Psi}$) can be incorporated in the set of consistent disturbance data \eqref{eq:distset_consistency}, and we give an example that illustrates the findings regarding sets of consistent disturbance data.

\subsubsection{Prior Model Knowledge}
Note that the system matrices $\bm{\Phi}$, $\bm{\Psi}$ of \eqref{eq:system} are related to the disturbance data $\glsd{dd}_{T}$ via \eqref{eq:connection_sysparams_data}. Suppose prior model knowledge for the system parameters $\bm{\Phi}$, $\bm{\Psi}$ is given in the following form with parameter matrices $\bm{G}_{\mathbb{A},1}$, $\bm{G}_{\mathbb{A},2}$, $\bm{G}_{\mathbb{A},3}$ (cf. \cite{alanwar2023data}):
\begin{equation} \label{eq:priormodelknowledge}
   \bm{G}_{\mathbb{A},1}\mat{\bm{\Phi} & \bm{\Psi}} \bm{G}_{\mathbb{A},2} \le \bm{G}_{\mathbb{A},3}.
\end{equation}
The set of system parameters \eqref{eq:priormodelknowledge} is unbounded in general (i.e., bounds are not given for all entries of $\bm{\Phi}$, $\bm{\Psi}$) or might even have an empty interior (i.e., entries of $\bm{\Phi}$, $\bm{\Psi}$ might be exactly known).
In order to incorporate prior model knowledge of the form \eqref{eq:priormodelknowledge} into the set of consistent disturbance data \eqref{eq:distset_consistency}, we first express \eqref{eq:connection_sysparams_data} in terms of $\glsd{dd}_{\glsd{xi}+\glsd{u}}$, i.e., the elements of the set of consistent disturbance data \eqref{eq:distset_consistency}. By substituting \eqref{eq:consistency_solutions} in \eqref{eq:connection_sysparams_data}, we retrieve
\begin{equation}
    \mat{\bm{\Phi} & \bm{\Psi}} = \bm{\Gamma}^{\mathbb{A}}_1 + \hankel{1}{\glsd{dd}_{\glsd{xi}+\glsd{u}}}\bm{\Gamma}^{\mathbb{A}}_2, \label{eq:connection_sysparams_data_reduced}
\end{equation}
with the purely data-dependent parameters 
\begin{align}
    \bm{\Gamma}^{\mathbb{A}}_1 = \left(\hankel{1}{\glsd{yd}_{T}} - \bm{\Gamma}_{1}\right) \bm{S}^{\dagger},~~
    \bm{\Gamma}^{\mathbb{A}}_2 = - \bm{\Gamma}_{2}\bm{S}^{\dagger}. 
\end{align}
Now, using \eqref{eq:connection_sysparams_data_reduced}, we can translate \eqref{eq:priormodelknowledge} into constraints for the disturbance data $\glsd{dd}_{\glsd{xi}+\glsd{u}}$, i.e., 
\begin{equation}
    \bm{G}_{\mathbb{A},1} \hankel{1}{\glsd{dd}_{\glsd{xi}+\glsd{u}}}\bm{\Gamma}^{\mathbb{A}}_2 \bm{G}_{\mathbb{A},2} \le \bm{G}_{\mathbb{A},3} - \bm{G}_{\mathbb{A},1}\bm{\Gamma}^{\mathbb{A}}_1 \bm{G}_{\mathbb{A},2}
\end{equation}
which can be incorporated as additional constraints into \eqref{eq:distset_consistency}. Note that the interior of \eqref{eq:distset_consistency} becomes empty if \eqref{eq:priormodelknowledge} has empty interior (the number of free parameters is reduced).

\subsubsection{Example}
Consider a system of the form \eqref{eq:system} with $\gls{Tini} = 1$ and the matrices $\bm{\Phi} = \mat{1 & 1}$, $\bm{\Psi} = 0$. From this system, a data trajectory $\{\gls{ud}_i\}_{i=0}^{T}$, $\{\gls{yd}_i\}_{i=0}^{T}$ of length $T=40$ is collected, satisfying Assumption~\ref{assum:trajData}, with inputs and disturbances randomly chosen within the bounds $\norminf{\gls{u}_k} \le 0.3$, $\norminf{\gls{d}_k} \le 0.1$ according to a uniform distribution. Using these data, the set of consistent disturbance data \eqref{eq:distset_consistency} is built (see Section~\ref{sec:distset_consistency}), which is illustrated in Fig.~\ref{fig:ex_cons}. It can be seen that the set of disturbance trajectories satisfying the consistency constraint~\eqref{eq:consistency} is remarkably smaller than the region that is defined by the original bound $\norminf{\gls{d}_k} \le 0.1$. As also depicted in Fig.~\ref{fig:ex_cons}, the size of the set can be further reduced by considering prior model knowledge $\bm{\Phi} \mat{0&1}^{\top} = 1$, $\bm{\Psi} = 0$.

\begin{figure}
    \centering
    \includegraphics[width=0.5\textwidth]{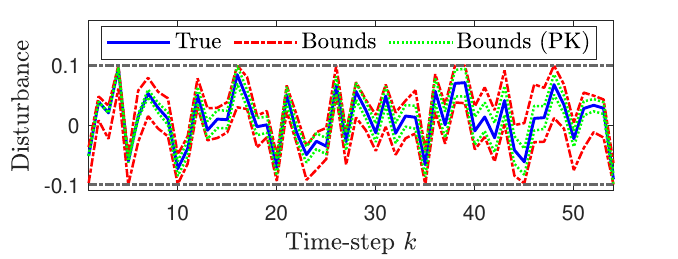}
    \caption{Visualization of the consistency constraint on the disturbance data. The original disturbance bound is shown in black, while the true (unknown) disturbance data are given in blue. Bounds resulting from the set of consistent disturbance data \eqref{eq:distset_consistency} are depicted in red. The bounds depicted in green additionally include prior model knowledge.}
    \label{fig:ex_cons}
\end{figure}


\subsection{Sampling-based Approximation of $\eps$-CSS} \label{app:samplingprelims}

Here, we briefly motivate and present two popular approaches and the relevant theory for sampling-based inner-approximations of $\eps$-CCSs. For a detailed discussion, see \cite{mammarella2022chance}.

When uncertainties propagate through the system dynamics nonlinearly (e.g., parametric uncertainty), or when uncertainties follow a non-Gaussian distribution, reformulating the stochastic chance constraints into tractable deterministic expressions is challenging. In such cases, sampling-based methods provide simple approaches for the deterministic approximation of chance constraints. A popular choice is known as the \textit{scenario} approach \cite{farina2016stochastic}, where the chance constraints are replaced by hard constraints that must be satisfied for a specified number of predicted sample trajectories, resulting from samples of the uncertainty drawn \textit{online} for every MPC iteration. Although the application is simple, the main disadvantages of scenario MPC are 1) high online computational complexity and 2) lack of closed-loop guarantees. To overcome these issues, \textit{offline}-sampling approaches have been proposed that aim to directly obtain a deterministic approximation of the chance constraints using samples of the uncertainty \cite{mammarella2022chance,lorenzen2017stochastic}. This allows for reduction of the online computational complexity of the controller, as well as closed-loop guarantees \cite{lorenzen2017stochastic,mammarella2018offline}.

Consider a general joint chance constraint $\Pr{\bm{G}_{\zeta}(\bm{w}) \bm{\zeta} \le \bm{g}_{\zeta}(\bm{w})} \ge 1-\eps$ where $\bm{\zeta}$ is the (deterministic) decision variable, and $\bm{G}_{\zeta}(\bm{w})\in\mathbb{R}^{n_{\mathrm{c}} \times n_{\zeta}}$, $\bm{g}_{\zeta}(\bm{w}) \in \mathbb{R}^{n_{\mathrm{c}}}$ are constraint parameters that depend on the realization $\bm{w} \in \mathbb{R}^{n_w}$ of a multivariate random variable. The corresponding $\eps$-CCS is defined as 
\begin{equation} \label{eq:chanceconsset}
    \mathbb{Z}^{\mathrm{P}} = \left\{ \bm{\zeta} \in \mathbb{R}^{n_{\zeta}} ~\left|~ \Pr{\bm{G}_{\zeta}(\bm{w}) \bm{\zeta} \le \bm{g}_{\zeta}(\bm{w})} \ge 1-\eps \right. \right\}.
\end{equation}
The goal of offline-sampling-based approaches is to determine a deterministic inner-approximation $\mathbb{Z}^{\mathrm{S}}$ of the $\eps$-CCS \eqref{eq:chanceconsset} by using $N_{\mathrm{s}}$ iid uncertainty samples $\bm{w}^{(i)}$, $i \in \mathbb{N}_1^{N_{\mathrm{s}}}$, yielding $\Pr{\mathbb{Z}^{\mathrm{S}} \subseteq \mathbb{Z}^{\mathrm{P}}} \ge 1-\conf$ with user-chosen confidence $\conf$.
The two popular approaches presented in the following, namely the direct sampling-based approximation and the probabilistic scaling approach, make use of the sampled set corresponding to \eqref{eq:chanceconsset}, i.e. (for a given sample $\bm{w}^{(i)}$)
\begin{equation} \label{eq:sampledconsset}
    \tilde{\mathbb{Z}}^{\mathrm{S}}\left(\bm{w}^{(i)}\right) = \left\{ \bm{\zeta} \in \mathbb{R}^{n_{\zeta}} \,\left|\, {\bm{G}_{\zeta}\left(\bm{w}^{(i)}\right)} \bm{\zeta} \le \bm{g}_{\zeta}\left(\bm{w}^{(i)}\right)\right. \right\}.
\end{equation}

\subsubsection{Direct Sampling-based Approximation} \label{sec:directsamplapprox}
For the direct sampling-based approximation of the $\eps$-CCS~\eqref{eq:chanceconsset}, the $N_{\mathrm{s}}$ samples of $\bm{w}$ and \eqref{eq:sampledconsset} are used to define the sampled set $\mathbb{Z}^{\mathrm{S}}_{\mathrm{LT}} \coloneqq \cap_{i=0}^{N_{\mathrm{s}}}\, \tilde{\mathbb{Z}}^{\mathrm{S}}\left(\bm{w}^{(i)}\right)$. 
The following result from statistical learning theory~\cite{alamo2009randomized} allows us to determine the required number of samples $N_{\mathrm{s}}$ (i.e., the sample complexity) for which the sampled set $\mathbb{Z}^{\mathrm{S}}_{\mathrm{LT}}$ is a subset of the $\eps$-CSS~\eqref{eq:chanceconsset} with a predefined confidence $\conf$. 
\begin{proposition}[Learning Theory Bound~\cite{mammarella2022chance}] \label{prop:sampling}
    For any risk parameter $\eps \in (0,0.14)$, confidence level $\conf \in (0,1)$, and sample complexity $N_{\mathrm{s}} \ge N_{\mathrm{LT}}(\eps,\,\conf,\,n_{\zeta},\,n_{\mathrm{c}})$  with
    \begin{equation} \label{eq:samplecomplexity}
       N_{\mathrm{LT}} \coloneqq \frac{4.1}{\eps}\left( \ln \frac{21.64}{\conf} + 4.39n_{\zeta} \log_2 \frac{8 e n_{\mathrm{c}}}{\eps}\right)
    \end{equation}
    and Euler's number $e$, it holds that $\Pr{\mathbb{Z}^{\mathrm{S}}_{\mathrm{LT}} \subseteq \mathbb{Z}^{\mathrm{P}}} \ge 1-\conf$.
\end{proposition}

Application of Proposition~\ref{prop:sampling} in the context of sampling-based Stochastic MPC was first presented in \cite{lorenzen2017stochastic} for the case of single chance constraints (i.e., $n_{\mathrm{c}} = 1$). A disadvantage of Proposition~\ref{prop:sampling} is that the sample complexity bound \eqref{eq:samplecomplexity} is rather conservative, easily leading to millions of sampled constraints even for small-scale systems \cite{mammarella2022chance}. Although the sampled constraints in \eqref{eq:sampledconsset} are generally highly redundant and can be reduced offline \cite{lorenzen2017stochastic}, the final number of constraints might still be too large to guarantee real-time implementability of the resulting predictive controller. For this reason, approaches have emerged that probabilistically scale a pre-defined set of fixed complexity to retrieve an inner-approximation of the $\eps$-CSS~\eqref{eq:chanceconsset}, as presented next.  

\subsubsection{Approximation via Probabilistic Scaling} \label{sec:probscalingapprox}
This approach is based on the idea of approximating the $\eps$-CSS \eqref{eq:chanceconsset} via a scalable simple approximating set (SAS) 
\begin{equation}
    \mathbb{Z}^{\mathrm{S}}(\sigma) \coloneqq \{\bm{\zeta}_{\mathrm{c}}\} \oplus \sigma \mathbb{Z}^{\mathrm{SAS}},
\end{equation}
with the center $\bm{\zeta}_{\mathrm{c}}$, the shape $\mathbb{Z}^{\mathrm{SAS}}$, and the scaling factor $\sigma \ge 0$. The designer controls the complexity of the approximating set by suitable choice of the design parameters $\bm{\zeta}_{\mathrm{c}}$, $\mathbb{Z}^{\mathrm{SAS}}$.

The goal of the probabilistic scaling approach is to find an optimal scaling factor $\sigma^*$ such that $\Pr{\mathbb{Z}^{\mathrm{S}}(\sigma^*) \subseteq \mathbb{Z}^{\mathrm{P}}} \ge 1-\conf$ with a desired level of confidence $\conf$ by exploiting samples. Before we can proivde the probabilistic scaling approach, consider the following definition of the scaling factor.
\begin{definition}[Scaling Factor~\cite{mammarella2022chance}] \label{def:scaling}
    For a given SAS $\mathbb{Z}^{\mathrm{S}}(\sigma)$ with center $\bm{\zeta}_{\mathrm{c}}$ and shape $\mathbb{Z}^{\mathrm{SAS}}$, and a sample $\bm{w}$, the scaling factor $\sigma(\bm{w})$ of $\mathbb{Z}^{\mathrm{S}}(\sigma)$ relative to $\bm{w}$ is defined as
    \begin{equation} \label{eq:scaling}
        \sigma(\bm{w}) \coloneqq \left\{ \begin{array}{cc}
            \max\limits_{\mathbb{Z}^{\mathrm{S}}(\sigma) \subseteq \tilde{\mathbb{Z}}^{\mathrm{S}}(\bm{w})} \sigma & \mathrm{if}~ \bm{\zeta}_{\mathrm{c}} \in \tilde{\mathbb{Z}}^{\mathrm{S}}(\bm{w}) \\
            0 & \mathrm{otherwise.}
        \end{array}\right.
    \end{equation}
\end{definition}
\begin{proposition}[Probabilistic Scaling of SAS~\cite{mammarella2022chance}] \label{prop:scaling}
    For a given candidate SAS $\mathbb{Z}^{\mathrm{S}}(\sigma)$ with center $\bm{\zeta}_{\mathrm{c}} \in \mathbb{Z}^{\mathrm{P}}$, any risk parameter $\eps \in (0,1)$, and confidence level $\conf \in (0,1)$, let the sample complexity $N_{\mathrm{s}}$ be chosen as $N_{\mathrm{s}} \ge N_{\mathrm{PS}}(\eps,\,\conf)$, with
     \begin{equation} \label{eq:samplecomplexity_scaling}
         N_{\mathrm{PS}} \coloneqq \frac{7.47}{\eps}\ln\frac{1}{\conf}.
     \end{equation}
     Furthermore, for $N_{\mathrm{s}}$ iid uncertainty samples $\bm{w}^{(i)}$, $i \in \mathbb{N}_1^{N_{\mathrm{s}}}$, let $\bm{\sigma}$, $[\bm{\sigma}]_{[i]} = \sigma\left(\bm{w}^{(i)}\right)$, be the vector of scaling factors determined via Definition~\ref{def:scaling}. Then, $\Pr{\mathbb{Z}^{\mathrm{S}}(\sigma^*) \subseteq \mathbb{Z}^{\mathrm{P}}} \ge 1-\conf$ holds, where $\sigma^*$ is the $N_{\mathrm{r}}$-th smallest entry of $\bm{\sigma}$, with the discarding parameter $N_{\mathrm{r}} = \lceil \frac{\eps N_{\mathrm{s}}}{2} \rceil$ and the ceil-function~$\lceil\cdot\rceil$.
\end{proposition}

The bound \eqref{eq:samplecomplexity_scaling} is independent from the number of constraints and dimension of decision variable, and thus lower than the bound defined in \eqref{eq:samplecomplexity} in most cases. Furthermore, the complexity of the inner-approximation $\mathbb{Z}^{\mathrm{S}}(\sigma^*)$ is fully determined by the shape $\mathbb{Z}^{\mathrm{SAS}}$. However, it is to note that the conservatism of Proposition~\ref{prop:scaling} depends on how well $\mathbb{Z}^{\mathrm{SAS}}$ captures the shape of the $\eps$-CSS \eqref{eq:chanceconsset}, and that for every uncertainty sample an optimization problem needs to be solved, see \eqref{eq:scaling}. Depending on the candidate SAS shape, this optimization problem might be computationally infeasible. A natural candidate for the SAS shape $\mathbb{Z}^{\mathrm{SAS}}$ is a sampled-polytope SAS $\tilde{\mathbb{Z}}^{\mathrm{SAS}}$, akin to \eqref{eq:sampledconsset}, constructed with a fixed number $\tilde{N}_{\mathrm{s}}$ of ``design" uncertainty samples $\tilde{\bm{w}}^{(i)}$, $i \in \mathbb{N}_1^{\tilde{N}_{\mathrm{s}}}$. Thus, the complexity of the shape can be determined apriori by $\tilde{N}_{\mathrm{s}}$. For such polytopic SAS, the optimization in \eqref{eq:scaling} can be done efficiently via linear programming.
For other possible choices of SAS shapes, we refer to the discussion in \cite{mammarella2022chance}.

\subsubsection{Approximation Procedures tailored to Section~\ref{sec:conssampling}}
\paragraph{Direct Approximation}
Draw $N_{\mathrm{s}} \ge N_{\mathrm{LT}}(\eps,\,\conf_y,\,n_{\zeta},\,n_{\mathrm{c},y})$ uncertainty samples $\bm{w}^{(i)}$, with $N_{\mathrm{LT}}$ from \eqref{eq:samplecomplexity} and $n_{\mathrm{c},y}$ is the number of output constraints in \eqref{eq:outputcons}.
By Proposition~\ref{prop:sampling}, the set $\mathbb{Y}^{\mathrm{S}}_l \coloneqq \cap_{i=0}^{N_{\mathrm{s}}}\tilde{\mathbb{Y}}^{\mathrm{S}}_l\left(\bm{w}^{(i)}\right)$ satisfies $\Pr{\mathbb{Y}^{\mathrm{S}}_l \subseteq \mathbb{Y}^{\mathrm{P}}_l} \ge 1-\conf_y$, thus retrieving a valid inner-approximation of \eqref{eq:ocp_original_chancecons} with confidence~$\conf_y$.
\paragraph{Probabilistic Scaling}
Draw $\tilde{N}_{\mathrm{s}}$ ``design" uncertainty samples $\bm{w}^{(i)}$ and construct a polytopic candidate scalable simple approximating set (SAS) $\mathbb{Y}^{\mathrm{S}}_l(\sigma) \coloneqq \{\bm{c}_{\mathrm{S}}\} \oplus \sigma \left(\mathbb{Y}^{\mathrm{SAS}}_l \ominus \left\{\bm{c}_{\mathrm{S}}\right\}\right)$ with scaling factor $\sigma > 0$, where $\mathbb{Y}^{\mathrm{SAS}}_l \coloneqq \cap_{i=0}^{\tilde{N}_{\mathrm{s}}}\tilde{\mathbb{Y}}^{\mathrm{S}}_l\left(\bm{w}^{(i)}\right)$ and $\bm{c}_{\mathrm{S}}$ is a center (e.g., Chebyshev or geometric center) of $\mathbb{Y}^{\mathrm{SAS}}_l$. Draw $N_{\mathrm{s}} \ge N_{\mathrm{PS}}(\eps,\,\conf_y)$ uncertainty samples with $N_{\mathrm{PS}}$ from \eqref{eq:samplecomplexity_scaling} and apply Proposition~\ref{prop:scaling} to retrieve the set $\mathbb{Y}^{\mathrm{S}}_l(\sigma^*)$ that satisfies $\Pr{\mathbb{Y}^{\mathrm{S}}_l(\sigma^*) \subseteq \mathbb{Y}^{\mathrm{P}}_l} \ge 1-\conf_y$, thus retrieving a valid inner-approximation of \eqref{eq:ocp_original_chancecons}  with confidence $\conf_y$.


\subsection{Data-driven Design of Terminal Ingredients} \label{app:stabilizingIngredients}
Here, we describe approaches on how to determine a stabilizing feedback gain $\bm{K}$, weighting matrix $\bm{P}$, and terminal set $\mathbb{X}_{\gls{Tf}}$ from Assumption~\ref{assum:stabilizingIngredients}, given data from Assumption~\ref{assum:trajData}.

\paragraph{Stabilizing feedback gain $\bm{K}$}
The feedback gain $\bm{K}$ should be chosen such that $\tilde{\bm{A}}_{\mathrm{cl},j} \coloneqq \tilde{\bm{A}}_j + \bm{K}\tilde{\bm{B}}_j$ is Schur stable for all $j \in \mathbb{N}_1^{N_{\mathrm{v}}}$, see Assumption~\ref{assum:stabilizingIngredients}(a). From \cite[Lemma~3]{pan2023data}, $\tilde{\bm{A}}_{\mathrm{cl},j}$ can be equivalently described as
\begin{equation} \label{eq:cldata}
    \tilde{\bm{A}}_{\mathrm{cl},j} = \bm{H}_{\xi,j}^{+} \bm{\Theta},
\end{equation}
with $\bm{H}_{\xi,j}^{+} \coloneqq \mat{\gls{xid}_{2} \, \cdots \, \gls{xid}_{T+1}} -  \tilde{\bm{E}} \bm{H}_1\left(\glsd{dd}_{T,j}\right)$ and for some $\bm{\Theta}$ that satisfies
\begin{equation}
    \mat{\bm{H}_1\left(\glsd{xid}_{T}\right) \\ \bm{H}_1\left(\glsd{ud}_{T}\right)} \bm{\Theta} = \mat{\bm{I}_{\glsd{xi}} \\ \bm{K}},
\end{equation}
where $\glsd{dd}_{T,j}$ is constructed using \eqref{eq:consistency_solutions} and the vertex $\glsd{dd}_{\glsd{xi}+\glsd{u},j}$ from the set of consistent disturbance data \eqref{eq:distset_consistency}. Following ideas from \cite{de2019formulas}, a feedback gain $\bm{K}$ that stabilizes all $\tilde{\bm{A}}_{\mathrm{cl},j}$, $j \in \mathbb{N}_1^{N_{\mathrm{v}}}$, can be found by solving the linear matrix inequalities
\begin{equation} \label{eq:determineK}
    \mat{\bm{H}_1\left(\glsd{xid}_{T}\right) \bm{\Theta} & \bm{H}_{\xi,j}^{+} \bm{\Theta} \\
    \left(\bm{H}_{\xi,j}^{+} \bm{\Theta}\right)^{\top} & \bm{H}_1\left(\glsd{xid}_{T}\right) \bm{\Theta}
    } \succ \bm{0}~~\forall j \in \mathbb{N}_1^{N_{\mathrm{v}}}
\end{equation}
for $\bm{\Theta}$, resulting in $\bm{K} = \bm{H}_1\left(\glsd{ud}_{T}\right) \bm{\Theta} \left(\bm{H}_1\left(\glsd{xid}_{T}\right) \bm{\Theta}\right)^{-1}$.

\paragraph{Weighting matrix $\bm{P}$}
The matrix $\bm{P}$ should be chosen such that Assumption~\ref{assum:stabilizingIngredients}(b) holds for given weights $\bm{Q}$, $\bm{R}$ and feedback gain $\bm{K}$. With the matrix vertices $\tilde{\bm{A}}_{\mathrm{cl},j}$, $j \in \mathbb{N}_1^{N_{\mathrm{v}}}$, we can find a matrix $\bm{P}$ that satisfies Assumption~\ref{assum:stabilizingIngredients}(b) by solving
\begin{subequations} \label{eq:determineP}
    \begin{align}
        & \underset{\tilde{\bm{P}}}{\mathrm{minimize}}  ~~\mathrm{trace}(\tilde{\bm{P}}) \\
        \mathrm{s.t.} & \mat{\tilde{\bm{P}} - \bm{Q}_{P} & \tilde{\bm{A}}_{\mathrm{cl},j}^{\top} \tilde{\bm{P}} \\
        \tilde{\bm{P}} \tilde{\bm{A}}_{\mathrm{cl},j} & \tilde{\bm{P}}
        } \succ \bm{0}~~\forall j \in \mathbb{N}_1^{N_{\mathrm{v}}},\\
        &~\tilde{\bm{P}} - \bm{Q}_{P} \succ \bm{0},
    \end{align}
\end{subequations}
with $\bm{Q}_{P} \coloneqq \bm{K}^{\top} \bm{R} \bm{K} +  \tilde{\bm{E}} \bm{Q} \tilde{\bm{E}}^{\top} $, $\bm{P} = \tilde{\bm{P}} - \tilde{\bm{E}} \bm{Q} \tilde{\bm{E}}^{\top}$. We remark that \eqref{eq:determineK} and \eqref{eq:determineP} can be simplified by over-approximating the set of matrix vertices \eqref{eq:setsysmat} via interval matrices, and then using the result from \cite{almao2007vertex} to reduce the number of to-be-checked vertices.

\paragraph{Terminal set $\mathbb{X}_{\gls{Tf}}$}
First, let us denote the constraint set of the extended state $\gls{xi}$ as $\mathbb{X} \coloneqq \left\{ \gls{xi} \in \mathbb{R}^{\glsd{xi}} ~\left|~ \bm{G}_{\xi} \gls{xi} \le \bm{g}_{\xi} \right.\right\}$. The set $\mathbb{X}$ can be constructed by considering the input and output constraints \eqref{eq:constraints} as well as the definition of the extended state in \eqref{eq:extstate}. The terminal set $\mathbb{X}_{\gls{Tf}}$ is defined as a subset of $\mathbb{X}$ that is RPI under the control law $\gls{u}_k = \bm{K} \gls{xi}_k$. By employing \cite[Section~5.3]{blanchini2015set}, we can make use of the given matrix vertices $\tilde{\bm{A}}_{\mathrm{cl},j}$, $j \in \mathbb{N}_1^{N_{\mathrm{v}}}$ and disturbance bound $\gls{dset}$ from \eqref{eq:distset} to determine $\mathbb{X}_{\gls{Tf}}$ through $\mathbb{X}_{\gls{Tf}} = \cap_{i=0}^{\infty} \mathbb{X}^{i}$, with $\mathbb{X}^{0} = \mathbb{X}$ and
\begin{equation} \label{eq:termset_invariant}
    \mathbb{X}^{i+1} = \left\{ \gls{xi} \in \mathbb{X}^{i} \left| \begin{array}{l} 
	\forall \gls{d} \in \gls{dset},\, j \in \mathbb{N}_1^{N_\mathrm{v}}: \\
	\bm{K}\gls{xi} \in \mathbb{U},~
    \tilde{\bm{A}}_{\mathrm{cl},j} \gls{xi} +  \tilde{\bm{E}} \gls{d} \in \mathbb{X}^{i}
	\end{array} \hspace{-2mm}\right.\right\}.
\end{equation}
In practice, the recursion \eqref{eq:termset_invariant} is terminated once $\mathbb{X}^{i+1} = \mathbb{X}^{i}$ for some $i \in \mathbb{N}$, yielding $\mathbb{X}_{\gls{Tf}} = \mathbb{X}^{i}$. Approximation techniques for termination after user-chosen finite iterations are described in \cite[Section~5.3]{blanchini2015set}. Note that, as $\gls{dset}$ is polytopic, it is sufficient to only take its vertices into account.

\subsection{Proof of Theorem~\ref{th:recfeas}} \label{app:recfeasproof}
\paragraph{Proof of Recursive Feasibility} By robustness of the first-step constraint \eqref{eq:constraints_firststep}, $\gls{var}_k \in \mathbb{C}_{\mathrm{R}}$ implies $\gls{xi}_{k+1} \in \gls{roa}$. By construction, it holds that $\gls{roa} \subset \left\{\gls{xi} \,\left|\, \mathbb{F}\left(\gls{xi}\right) \neq \emptyset \right.\right\}$, which proves the claim.

\paragraph{Proof of Closed-loop Constraint Satisfaction} With $\gls{xi}_0 \in \gls{roa}$, a feasible pair $\col{\gls{xi}_0,\, \gls{v}_{\mathrm{f},0},\,\tilde{\bm{\alpha}}_0} \in \mathbb{C}$ exists by design. Closed-loop input constraint satisfaction follows from recursive feasibility and the constraint $\gls{v}_{0|k} + \bm{K} \gls{xi}_k \in \tilde{\mathbb{U}}_0 = \mathbb{U}$, which is included in the constraint set $\mathbb{C}$.
Furthermore, it holds that $\mathbb{C} \subseteq \tilde{\mathbb{Y}}_0$, and $\tilde{\mathbb{Y}}_0 \subseteq \mathbb{Y}_0^P$ with confidence $1-\conf_y$ by design (see Section~\ref{sec:conssampling}). Thus, the chance constraint $\Pr{\bm{G}_y \gls{y}_{0|k} \le \bm{g}_y} \ge 1-\eps$ is satisfied with confidence $1-\conf_y$ for all feasible $\gls{var}_k \in \mathbb{C}$, $k \ge 0$, which is sufficient for satisfaction of chance constraint \eqref{eq:outputcons} in closed-loop.

\subsection{Proof of Lemma~\ref{lem:saabound}} \label{app:saabound}
Let $\svmax{\cdot}$ denote the largest singular value of a matrix. With $\hat{J}^*(\gls{xi}_k) = \hat{J}(\gls{var}^*_k) $ and the definition of the cost functions in \eqref{eq:cost_exact} and \eqref{eq:cost_sampleaverage}, we have
    \begin{align*}
        \left\lvert \hat{J}(\gls{var}^*_k) -  J(\gls{var}^*_k)\right\rvert  &= \left\lvert \norm{\gls{var}^*_k}^2_{\hat{\bm{Q}}_{\mathrm{OCP}}} - \norm{\gls{var}^*_k}^2_{\E{\bm{Q}_{\mathrm{OCP}}(\bm{w})}} \right\rvert\\
        &= \left\lvert {\gls{var}^*_k}^{\top}\left(\hat{\bm{Q}}_{\mathrm{OCP}} - \E{\bm{Q}_{\mathrm{OCP}}(\bm{w})}\right)\gls{var}^*_k \right\rvert \\
        & \le \svmax{\hat{\bm{Q}}_{\mathrm{OCP}} - \E{\bm{Q}_{\mathrm{OCP}}(\bm{w})}} {\gls{var}^*_k}^{\top}\gls{var}^*_k  \\
        & \le \svmax{\hat{\bm{Q}}_{\mathrm{OCP}} - \E{\bm{Q}_{\mathrm{OCP}}(\bm{w})}} \norm{\gls{xi}_k}^2_{\bm{P}_{\mathrm{c}}}, 
    \end{align*}
    where $\norm{\gls{xi}_k}^2_{\bm{P}_{\mathrm{c}}} \ge \norm{\gls{var}^*_k}^2 = {\gls{var}^*_k}^{\top}\gls{var}^*_k $ is used for the last step. By definition, $\hat{\bm{Q}}_{\mathrm{OCP}} = (1/\gls{Nsaa})\sum_{i=1}^{\gls{Nsaa}} \bm{Q}_{\mathrm{OCP}}\left(\bm{w}^{(i)}\right)$ is the average of $\gls{Nsaa}$ independent samples $\bm{Q}_{\mathrm{OCP}}\left(\bm{w}^{(i)}\right)$. Leveraging the upper bound $\lambda_Q \ge \max_{\bm{w}} \evmax{\bm{Q}_{\mathrm{OCP}}(\bm{w})}$, the matrix Hoeffding inequality \cite[Thm.~6.15]{wainwright2019high} yields
    \begin{align*} 
        \Pr{\svmax{\hat{\bm{Q}}_{\mathrm{OCP}} - \E{\bm{Q}_{\mathrm{OCP}}(\bm{w})}} \le \tilde{\tau}}\hspace{25mm}&\\
        \ge 1- 2\glsd{var}\exp\left(-\frac{\gls{Nsaa}\tilde{\tau}^2}{2\lambda_Q^2}\right)&
    \end{align*}
    for arbitrary $\tilde{\tau} > 0$. The assertion then follows from introducing the confidence $\delta_{\mathrm{s}}$ and choosing $\tilde{\tau} =\tau\left(\gls{Nsaa},\,\delta_{\mathrm{s}}\right)$.
    
\subsection{Proof of Theorem~\ref{th:rasie}} \label{app:stabproof}
In order to prove RASiE, we construct a stochastic ISS Lyapunov function based on the optimal cost \eqref{eq:ocp}, and then employ the following proposition.
\begin{proposition}[Stochastic ISS Lyapunov Function \cite{mcallister2022nonlinear}] \label{prop:siss}
    The origin of the closed-loop system \eqref{eq:system_cl_nonminimal} is RASiE on the RPI set $\gls{roa}$ if there exists a function $V: \gls{roa} \rightarrow \mathbb{R}_{\ge 0}$ and functions $\rho_1$, $\rho_2$, $\rho_3 \in \mathscr{K}_{\infty}$, $\rho_4$, $\rho_5 \in \mathscr{K}$ such that  for all $\gls{xi}_{k} \in \gls{roa}$, it holds that
    \begin{subequations}
        \begin{align}
            \rho_1\left(\norm{\gls{xi}_k}\right) \le &~V\left(\gls{xi}_{k}\right) \le \rho_2\left(\norm{\gls{xi}_k}\right) + \rho_4\left(\trace{\bm{\Sigma}}\right), \label{eq:siss_bound}\\
            \E{V\left(\gls{xi}_{k+1}\right)} - &~V\left(\gls{xi}_{k}\right) \le - \rho_3\left(\norm{\gls{xi}_k}\right) + \rho_5\left(\trace{\bm{\Sigma}}\right). \label{eq:siss_descent}
        \end{align}
    \end{subequations}
    Function $V$ is then called a stochastic ISS Lyapunov function.
\end{proposition}

Let us first consider the case where the candidate solution \eqref{eq:candsol} for time-step $k+1$ is feasible. Using Lemma~\ref{lem:saabound} and the expected cost \eqref{eq:cost_exact}, we can bound the expected increase of the optimal SAA cost \eqref{eq:ocp} with a probability of at least $\delta_{\mathrm{S}}$ uniformly for all $\gls{xi}_k \in \gls{roa}$ (cf. Lemma~\ref{lem:saabound}) as
\begin{align} \label{eq:cost_increase}
    &\E{\hat{J}^*\left(\gls{xi}_{k+1}\right)~\big|~\tilde{\gls{v}}_{\mathrm{f},k+1}~\mathrm{feasible}} - \hat{J}^*\left(\gls{xi}_k\right) \notag \\
    & \le  \E{\hat{J}\left(\gls{xi}_{k+1},\,\tilde{\gls{v}}_{\mathrm{f},k+1}\right)} - \hat{J}^*\left(\gls{xi}_k\right) \notag\\
    & \overset{\eqref{eq:saacostbound}}{\le}  \E{J\left(\gls{xi}_{k+1},\,\tilde{\gls{v}}_{\mathrm{f},k+1}\right)} - J\left(\gls{xi}_k,\,\gls{v}^*_{\mathrm{f},k}\right) \notag\\
    &~~+ \tau \E{\max\limits_{\left[\tilde{\bm{A}} ~\tilde{\bm{B}}\right] \in \tilde{\mathbb{A}}} \norm{\tilde{\bm{A}}\gls{xi}_k + \tilde{\bm{B}} \gls{u}_k + \tilde{\bm{E}}\gls{d}_k}^2_{\bm{P}_{\mathrm{c}}}} + \tau \norm{\gls{xi}_k}^2_{\bm{P}_{\mathrm{c}}} \notag \\ 
    & \le  \E{J\left(\gls{xi}_{k+1},\,\tilde{\gls{v}}_{\mathrm{f},k+1}\right)} - J\left(\gls{xi}_k,\,\gls{v}^*_{\mathrm{f},k}\right) + \tau \E{\norm{\tilde{\bm{E}}\gls{d}_k}^2_{\bm{P}_{\mathrm{c}}}} \notag\\
    &~~+ \tau \max\limits_{\left[\tilde{\bm{A}} ~\tilde{\bm{B}}\right] \in \tilde{\mathbb{A}}} \norm{\mat{\tilde{\bm{A}} ~~ \tilde{\bm{B}}}\mat{\gls{xi}_k \\ \gls{u}_k}}^2_{\bm{P}_{\mathrm{c}}} + \tau \norm{\gls{xi}_k}^2_{\bm{P}_{\mathrm{c}}}
\end{align}
For the increase of the expected cost \eqref{eq:cost_exact} in \eqref{eq:cost_increase}, we retrieve
\begin{align}
    &  \E{J\left(\gls{xi}_{k+1},\,\tilde{\gls{v}}_{\mathrm{f},k+1}\right)} - J\left(\gls{xi}_k,\,\gls{v}^*_{\mathrm{f},k}\right) \notag\\
    &= \mathrm{E}\left[\sum\limits_{l=1}^{\gls{Tf}} \left( \norm{\gls{y}^*_{l|k}+\tilde{\bm{E}}^{\top}\tilde{\bm{A}}_{\mathrm{cl}}^l\tilde{\bm{E}}\gls{d}_k}^2_{\bm{Q}} \right. \right. \notag \\ 
    &~~+ \left. \left.\norm{\gls{u}^*_{l|k} +\bm{K}\tilde{\bm{A}}_{\mathrm{cl}}^{l-1}\tilde{\bm{E}}\gls{d}_k}^2_{\bm{R}} \right) + \norm{\tilde{\bm{A}}_{\mathrm{cl}}\gls{xi}^*_{\gls{Tf}|k} + \tilde{\bm{A}}_{\mathrm{cl}}^{\gls{Tf}}\tilde{\bm{E}}\gls{d}_k}^2_{\bm{P}} \right]\notag\\
    &~~ - \E{\sum\limits_{l=0}^{\gls{Tf}-1} \left( \norm{\gls{y}^*_{l|k}}^2_{\bm{Q}} + \norm{\gls{u}^*_{l|k}}^2_{\bm{R}} \right) + \norm{\gls{xi}^*_{\gls{Tf}|k}}^2_{\bm{P}}} \notag\\
    & \le \mathrm{E}\left[\sum\limits_{l=1}^{\gls{Tf}} \left( \norm{\tilde{\bm{E}}^{\top}\tilde{\bm{A}}_{\mathrm{cl}}^{l}\tilde{\bm{E}}\gls{d}_k}^2_{\bm{Q}} + \norm{\bm{K}\tilde{\bm{A}}_{\mathrm{cl}}^{l-1}\tilde{\bm{E}}\gls{d}_k}^2_{\bm{R}} \right) \right.\notag\\
    &~~ \left. + \norm{\tilde{\bm{A}}_{\mathrm{cl}}^{\gls{Tf}}\tilde{\bm{E}}\gls{d}_k}^2_{\bm{P}} \right] - \E{\norm{\gls{y}^*_{0|k}}^2_{\bm{Q}} + \norm{\gls{u}^*_{0|k}}^2_{\bm{R}} + \norm{\gls{xi}^*_{\gls{Tf}|k}}^2_{\bm{P}}} \notag \\
    &~~ + \E{\norm{\gls{y}^*_{\gls{Tf}|k}}^2_{\bm{Q}} + \norm{\gls{u}^*_{\gls{Tf}|k}}^2_{\bm{R}} + \norm{\tilde{\bm{A}}_{\mathrm{cl}}\gls{xi}^*_{\gls{Tf}|k}}^2_{\bm{P}}}\notag \\
    & \le \E{\norm{\tilde{\bm{E}}\gls{d}_k}^2_{\bm{P}}} - \E{\norm{\gls{y}^*_{0|k}}^2_{\bm{Q}} + \norm{\gls{u}^*_{0|k}}^2_{\bm{R}} + \norm{\gls{xi}^*_{\gls{Tf}|k}}^2_{\bm{P}}} \notag \\
    &~~ + \E{\norm{\tilde{\bm{E}}^{\top}\tilde{\bm{A}}_{\mathrm{cl}}\gls{xi}^*_{\gls{Tf}|k}}^2_{\bm{Q}} + \norm{\bm{K}\gls{xi}^*_{\gls{Tf}|k}}^2_{\bm{R}} + \norm{\tilde{\bm{A}}_{\mathrm{cl}}\gls{xi}^*_{\gls{Tf}|k}}^2_{\bm{P}}}\hspace{-5mm}\notag \\
    & \le \E{\norm{\tilde{\bm{E}}\gls{d}_k}^2_{\bm{P}} - \norm{\gls{y}^*_{0|k}}^2_{\bm{Q}} - \norm{\gls{u}_k}^2_{\bm{R}}}, 
\end{align}
using the Cauchy-Schwarz inequality and Assumption~\ref{assum:stabilizingIngredients}(b).

Now, for the case where the candidate solution \eqref{eq:candsol} for time-step $k+1$ is infeasible, we bound the expected cost increase of the optimal SAA cost \eqref{eq:ocp} using \eqref{eq:costbound} as
\begin{align*}
    &\E{\hat{J}^*\left(\gls{xi}_{k+1}\right)~\big|~\tilde{\gls{v}}_{\mathrm{f},k+1}~\mathrm{infeasible}} - \hat{J}^*\left(\gls{xi}_k\right)\notag\\
    &\le \E{\max\limits_{\left[\tilde{\bm{A}} ~\tilde{\bm{B}}\right] \in \tilde{\mathbb{A}}} \norm{\tilde{\bm{A}}\gls{xi}_k + \tilde{\bm{B}} \gls{u}_k + \tilde{\bm{E}}\gls{d}_k}^2_{\bm{P}_{\mathrm{u}}} - \norm{\gls{xi}_k}^2_{\bm{P}_{\mathrm{l}}}} \notag\\
    &\le \E{\max\limits_{\left[\tilde{\bm{A}} ~\tilde{\bm{B}}\right] \in \tilde{\mathbb{A}}} \norm{\mat{\tilde{\bm{A}} ~~ \tilde{\bm{B}}}\mat{\gls{xi}_k \\ \gls{u}_k}}^2_{\bm{P}_{\mathrm{u}}} \hspace{-1.5mm}- \norm{\gls{xi}_k}^2_{\bm{P}_{\mathrm{l}}} + \norm{\tilde{\bm{E}}\gls{d}_k}^2_{\bm{P}_{\mathrm{u}}}},
\end{align*}
Thus, applying the law of total probability, we obtain
\begin{align}
    &\E{\hat{J}^*\left(\gls{xi}_{k+1}\right)} - \hat{J}^*\left(\gls{xi}_k\right)\notag\\
    &\le \left(1-\eps_{\mathrm{f}}\right) \E{\norm{\tilde{\bm{E}}\gls{d}_k}^2_{\bm{P}+\tau\bm{P}_{\mathrm{c}}} - \norm{\gls{y}^*_{0|k}}^2_{\bm{Q}} - \norm{\gls{u}_k}^2_{\bm{R}}} \notag\\
    &+ \left(1-\eps_{\mathrm{f}}\right)\tau \left( \max\limits_{\left[\tilde{\bm{A}} ~\tilde{\bm{B}}\right] \in \tilde{\mathbb{A}}} \norm{\mat{\tilde{\bm{A}} ~~ \tilde{\bm{B}}}\mat{\gls{xi}_k \\ \gls{u}_k}}^2_{\bm{P}_{\mathrm{c}}} + \norm{\gls{xi}_k}^2_{\bm{P}_{\mathrm{c}}}\right) \notag \\
    &+ \eps_{\mathrm{f}}\E{\max\limits_{\left[\tilde{\bm{A}} ~\tilde{\bm{B}}\right] \in \tilde{\mathbb{A}}} \norm{\mat{\tilde{\bm{A}} ~~ \tilde{\bm{B}}}\mat{\gls{xi}_k \\ \gls{u}_k}}^2_{\bm{P}_{\mathrm{u}}} - \norm{\gls{xi}_k}^2_{\bm{P}_{\mathrm{l}}} + \norm{\tilde{\bm{E}}\gls{d}_k}^2_{\bm{P}_{\mathrm{u}}}}. \label{eq:costdescent}
\end{align}

Now, as a stochastic ISS Lyapunov function candidate, we define $V\left(\gls{xi}_k\right) \coloneqq \hat{J}^*\left(\gls{xi}_k\right) + \left(1-\eps_{\mathrm{f}}\right)c_{\mathrm{S}} W\left(\gls{xi}_k\right)$. By combining \eqref{eq:IOSSdescent} and \eqref{eq:costdescent}, we retrieve for the expected descent of the Lyapunov function candidate
\begin{align}
    &\E{V\left(\gls{xi}_{k+1}\right)} - V\left(\gls{xi}_k\right)\notag\\
    &\le \left(1-\eps_{\mathrm{f}}\right) \E{\norm{\tilde{\bm{E}}\gls{d}_k}^2_{\bm{P}+\tau\bm{P}_{\mathrm{c}}} - \norm{\gls{y}^*_{0|k}}^2_{\bm{Q}} - \norm{\gls{u}_k}^2_{\bm{R}}} \notag \\
    &~\,+ \eps_{\mathrm{f}}\E{\max\limits_{\left[\tilde{\bm{A}} ~\tilde{\bm{B}}\right] \in \tilde{\mathbb{A}}} \norm{\mat{\tilde{\bm{A}} ~~ \tilde{\bm{B}}}\mat{\gls{xi}_k \\ \gls{u}_k}}^2_{\bm{P}_{\mathrm{u}}} - \norm{\gls{xi}_k}^2_{\bm{P}_{\mathrm{l}}} + \norm{\tilde{\bm{E}}\gls{d}_k}^2_{\bm{P}_{\mathrm{u}}}}  \notag\\
    &~~+ \left(1-\eps_{\mathrm{f}}\right)\tau \left( \max\limits_{\left[\tilde{\bm{A}} ~\tilde{\bm{B}}\right] \in \tilde{\mathbb{A}}} \norm{\mat{\tilde{\bm{A}} ~~ \tilde{\bm{B}}}\mat{\gls{xi}_k \\ \gls{u}_k}}^2_{\bm{P}_{\mathrm{c}}} + \norm{\gls{xi}_k}^2_{\bm{P}_{\mathrm{c}}}\right)\notag \\
    &~~ + \left(1-\eps_{\mathrm{f}}\right)\E{c_{\mathrm{S}}c_u \norm{\gls{u}_k}^2 + c_{\mathrm{S}}c_y \norm{\gls{y}_k}^2 + c_{\mathrm{S}}c_d \norm{\gls{d}_k}^2} \notag \\
    &~~ -\left(1-\eps_{\mathrm{f}}\right)\frac{c_{\mathrm{S}}}{2}\norm{\gls{xi}_k}^2 \notag \\
    &\le \E{\norm{\gls{d}_k}^2_{\bm{P}_d}} -\left(1-\eps_{\mathrm{f}}\right)\left(\frac{c_{\mathrm{S}}}{2}\norm{\gls{xi}_k}^2 + \norm{\gls{u}_k}^2_{\bm{R}-c_{\mathrm{S}}c_u\bm{I}_{\glsd{u}}}  \phantom{\norm{\mat{\gls{xi}_k \\ \gls{u}_k}}^2_{\bm{P}_{\mathrm{u}}}}\right. \notag \\
    &~~ + \tau \left( \max\limits_{\left[\tilde{\bm{A}} ~\tilde{\bm{B}}\right] \in \tilde{\mathbb{A}}} \norm{\mat{\tilde{\bm{A}} ~~ \tilde{\bm{B}}}\mat{\gls{xi}_k \\ \gls{u}_k}}^2_{\bm{P}_{\mathrm{c}}} + \norm{\gls{xi}_k}^2_{\bm{P}_{\mathrm{c}}}\right)\notag \\
    &~~ \left. - \frac{\eps_{\mathrm{f}}}{1-\eps_{\mathrm{f}}} \left(\max\limits_{\left[\tilde{\bm{A}} ~\tilde{\bm{B}}\right] \in \tilde{\mathbb{A}}} \norm{\mat{\tilde{\bm{A}} ~~ \tilde{\bm{B}}}\mat{\gls{xi}_k \\ \gls{u}_k}}^2_{\bm{P}_{\mathrm{u}}}- \norm{\gls{xi}_k}^2_{\bm{P}_{\mathrm{l}}}\right)\right)  \notag\\ 
    &\le -\left(1-\eps_{\mathrm{f}}\right) \min_j\left(\evmin{\bm{Q}^{\mathrm{S}}_j}\right) \norm{\gls{xi}_k}^2 + \evmax{\bm{P}_d}\trace{\bm{\Sigma}},
\end{align}
with $\bm{P}_d \coloneqq \tilde{\bm{E}}^{\top}\left(\left(1-\eps_{\mathrm{f}}\right)\left(\bm{P}+\tau\bm{P}_{\mathrm{c}}\right) + \eps_{\mathrm{f}}\bm{P}_{\mathrm{u}}\right)\tilde{\bm{E}} + c_{\mathrm{S}}c_d\bm{I}_p$ and $\bm{Q}^{\mathrm{S}}_j$ from \eqref{eq:stabcond}. If $\bm{Q}^{\mathrm{S}}_j \succ \bm{0}$ for all $\mat{\tilde{\bm{A}}_j ~ \tilde{\bm{B}}_j} \in \tilde{\mathbb{A}}$, property \eqref{eq:siss_descent} is satisfied with a probability of at least $\delta_{\mathrm{S}}$ uniformly for all $\gls{xi}_k \in \gls{roa}$ (cf. Lemma~\ref{lem:saabound}) for $\rho_3\left(\norm{\gls{xi}_k}\right) = \left(1-\eps_{\mathrm{f}}\right) \min_j\left(\evmin{\bm{Q}^{\mathrm{S}}_j}\right) \norm{\gls{xi}_k}^2 $ and $\rho_5\left(\trace{\bm{\Sigma}}\right) = \evmax{\bm{P}_d}\trace{\bm{\Sigma}}$. Furthermore, property \eqref{eq:siss_bound} holds for $\rho_1\left(\norm{\gls{xi}_k}\right) = \norm{\gls{xi}_k}_{\bm{P}_{\mathrm{l}}}^2$, $\rho_2\left(\norm{\gls{xi}_k}\right) = \norm{\gls{xi}_k}_{\bm{P}_{\mathrm{u}}+\bm{P}_W}^2$, and arbitrary $\rho_4 \in \mathscr{K}$. Therefore, with confidence $\delta_{\mathrm{S}}$, $V\left(\gls{xi}_k\right)$ is a valid stochastic ISS Lyapunov function on $\gls{roa}$, and the assertion follows from Proposition~\ref{prop:siss}.

\section*{Acknowledgment}

The authors express their gratitude to Julian Berberich for valuable discussions and feedback.

\section*{References} \vspace{-2em}
\bibliographystyle{ieeetr}
\bibliography{mybib}

\end{document}